\documentclass[11pt]{article}

\usepackage{fullpage}

\title{Universal Hinge Patterns for Folding Strips \\ Efficiently into Any Grid Polyhedron}
\author{%
  Nadia M. Benbernou%
    \thanks{Google Inc., \protect\url{nbenbern@gmail.com}.
      Work performed while at MIT.}
\and
  Erik D. Demaine%
    \thanks{MIT Computer Science and Artificial Intelligence Laboratory,
      32 Vassar St., Cambridge, MA 02139, USA,
      \protect\url{{edemaine,mdemaine}@mit.edu}}
\and
  Martin L. Demaine\footnotemark[2]
\and
  Anna Lubiw%
    \thanks{David R. Cheriton School of Computer Science,
      University of Waterloo, Waterloo, Ontario N2L 3G1, Canada,
      \protect\url{alubiw@uwaterloo.ca}}
}
\date{}

\usepackage{amsmath}
\usepackage{amsthm}
\usepackage{amsfonts}
\usepackage{graphicx}
\usepackage{hyperref}
\usepackage{subfigure}
\usepackage{tabularx}
\def\min{\text{min}}
\def\max{\text{max}}

\def\T{\mathcal{T}}
\def\weight{\text{weight}}
\newtheorem{theorem}{Theorem}
\newtheorem{lemma}[theorem]{Lemma}
\newtheorem{prop}[theorem]{Proposition}
\theoremstyle{definition}

\newtheorem{corollary}[theorem]{Corollary}

{\makeatletter
 \gdef\xxxmark{%
   \expandafter\ifx\csname @mpargs\endcsname\relax 
     \expandafter\ifx\csname @captype\endcsname\relax 
       \marginpar{xxx}
     \else
       xxx 
     \fi
   \else
     xxx 
   \fi}
 \gdef\xxx{\@ifnextchar[\xxx@lab\xxx@nolab}
 \long\gdef\xxx@lab[#1]#2{{\bf [\xxxmark #2 ---{\sc #1}]}}
 \long\gdef\xxx@nolab#1{{\bf [\xxxmark #1]}}
}

\begin{document}
\maketitle

\begin{abstract}
  We present two universal hinge patterns that enable a strip of material to
  fold into any connected surface made up of unit squares on the 3D cube
  grid---for example, the surface of any polycube.
  The folding is efficient:
  for target surfaces topologically equivalent to a sphere,
  the strip needs to have only twice the target surface area,
  and the folding stacks at most two layers of material anywhere.
  These geometric results offer a new way to build programmable matter
  that is substantially more efficient than what is possible with
  a square $N \times N$ sheet of material, which can fold into all polycubes
  only of surface area $O(N)$ and may stack $\Theta(N^2)$ layers at one point.
  We also show how our strip foldings can be executed by a rigid motion
  without collisions (albeit assuming zero thickness),
  which is not possible in general with 2D sheet folding.

  To achieve these results, we develop new approximation algorithms
  for milling the surface of a grid polyhedron,
  which simultaneously give a $2$-approximation in tour length
  and an $8/3$-approximation in the number of turns.
  Both length and turns consume area when folding a strip,
  so we build on past approximation algorithms
  for these two objectives from 2D milling.
\end{abstract}

\section{Introduction}
\label{StripFolding:sec:intro}

In \emph{computational origami design}, the goal is generally to develop an
algorithm that, given a desired shape or property, produces a crease pattern
that folds into an origami with that shape or property.
Examples include folding any shape \cite{Demaine-Demaine-Mitchell-2000},
folding approximately any shape while being watertight
\cite{Demaine-Tachi-2017},
and optimally folding a shape whose projection is a desired metric tree
\cite{Lang-1996,Lang-Demaine-2006}.
In all of these results, every different shape or tree results in a
completely different crease pattern; two shapes rarely share many
(or even any) creases.

The idea of a \emph{universal hinge pattern} \cite{origami5}
is that a finite set of \emph{hinges} (possible creases)
suffice to make exponentially many different shapes.
The main result along these lines is that an $N \times N$
``box-pleat'' grid suffices to make any polycube made of $O(N)$ cubes
\cite{origami5}.  The \emph{box-pleat grid} is a square grid plus
alternating diagonals in the squares, also known as the ``tetrakis tiling''.
For each target polycube, a subset of the hinges in the grid
serve as the crease pattern for that shape.
Polycubes form a \emph{universal} set of shapes in that they can arbitrarily
closely approximate (in the sense of Hausdorff metric) any desired volume.

The motivation for universal hinge patterns is the implementation of
\emph{programmable matter}---material whose shape can be externally programmed.
One approach to programmable matter, developed by an MIT--Harvard
collaboration, is a \emph{self-folding sheet}---a sheet of material that can
fold itself into several different origami designs, without manipulation
by a human origamist \cite{pnas, robotica}. 
For practicality, the sheet must consist of a fixed pattern of hinges,
each with an embedded actuator that can be programmed to fold or not.
Thus for the programmable matter to be able to form a universal set of
shapes, we need a universal hinge pattern.

The box-pleated polycube result \cite{origami5}, however, has some practical
limitations that prevent direct application to programmable matter.
Specifically, using a sheet of area $\Theta(N^2)$ to fold $N$ cubes means
that all but a $\Theta(1/N)$ fraction of the surface area is wasted.
Unfortunately, this reduction in surface area is necessary for a roughly
square sheet, as folding a $1 \times 1 \times N$ tube requires a sheet
of diameter $\Omega(N)$.
Furthermore, a polycube made from $N$ cubes can have surface area as low as
$\Theta(N^{2/3})$, resulting in further wastage of surface area in the worst case.
Given the factor-$\Omega(N)$ reduction in surface area, an average of
$\Omega(N)$ layers of material come together on the polycube surface.
Indeed, the current approach can have up to $\Theta(N^2)$ layers
coming together at a single point \cite{origami5}.
Real-world robotic materials have significant thickness,
given the embedded actuation and electronics,
meaning that only a few overlapping layers are really practical \cite{pnas}.

\paragraph{Our results: strip folding.}
In this paper, we introduce two new universal hinge patterns that avoid
these inefficiencies, by using sheets of material that are long only in one
dimension (``strips'').
Specifically, Figure~\ref{StripFolding:fig:StripTypes} shows the two hinge patterns:
the \emph{canonical strip} is a $1 \times N$ strip with hinges at integer
grid lines and same-oriented diagonals, while the \emph{zig-zag strip} is an
$N$-square zig-zag with hinges at just integer grid lines.
We show in Section~\ref{Universality}
that any \emph{grid surface}---any connected surface made up of
unit squares on the 3D cube grid%
%
---can be folded from either strip.  The strip length only needs to be a
constant factor larger than the surface area, and the number of layers is
at most a constant throughout the folding.
Most of our analysis concerns (genus-0) \emph{grid polyhedra}, that is,
when the surface is topologically equivalent to a sphere
(a manifold without boundary, so that every edge is incident to exactly
two grid squares, and without handles, unlike a torus).
We show in Section~\ref{StripFolding:sec:StripFolding} that
a grid polyhedron of surface area $N$ can be folded from a
canonical strip of length $2 N$ with at most two layers everywhere,
or from a zig-zag strip of length $4 N$ with at most four layers everywhere.

\begin{figure}
\centering
\subfigure[\label{StripFolding:fig:CanonicalStrip}]{\includegraphics[scale=0.5]{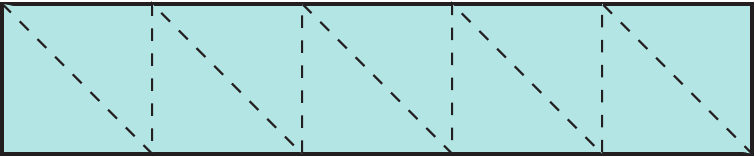}}\hfil\hfil\hfil
 \subfigure[\label{StripFolding:fig:ZigZagStrip}]{\includegraphics[scale=0.5]{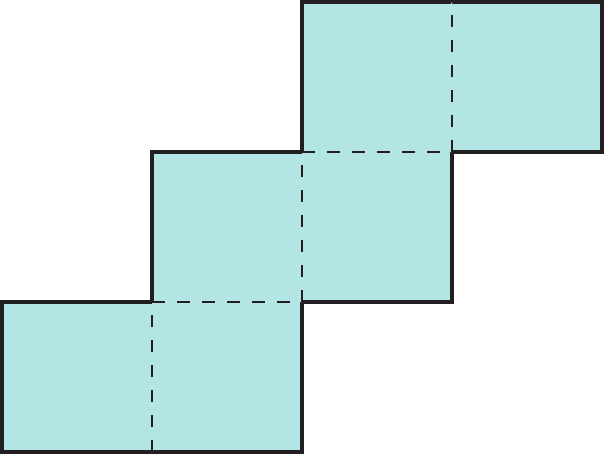}}
\caption{Two universal hinge patterns in strips. (a) A canonical strip of length $5$. (b) A zig-zag strip of length~$6$.
 The dashed lines are hinges.}
\label{StripFolding:fig:StripTypes}
\end{figure}

The improved surface efficiency and reduced layering of these strip results
seem more practical for programmable matter.
In addition, the panels of either strip (the facets delineated by hinges)
are connected acyclically into a path, making them potentially
easier to control.
One potential drawback is that the reduced connectivity
makes for a flimsier device; this issue can be mitigated by
adding tabs to the edges of the strips to make full two-dimensional contacts
across seams and thereby increase strength.

We also show in Section~\ref{StripFolding:subsec:Collision}
an important practical result for our strip foldings:
assuming a small lower bound on feature size,
we give an algorithm for actually folding the strip into the desired shape,
while keeping the panels rigid (flat) and avoiding self-intersection
throughout the motion.
Such a rigid folding process is important given current fabrication materials,
which put flexibility only in the creases between panels \cite{pnas}.
An important limitation, however, is that we assume zero thickness
of the material, which would need to be avoided before this method
becomes practical.

Our approach is also related to the 1D chain robots of \cite{Moteins_TRO},
but based on thin material instead of thick solid chains.
Most notably, working with thin material enables us to use a few
overlapping layers to make any desired surface without scaling,
and still with high efficiency.
Essentially, folding long thin strips of sheet material is like a fusion
between 1D chains of \cite{Moteins_TRO} and the square sheet folding of
\cite{origami5,pnas,robotica}.

\paragraph{Milling tours.}

At the core of our efficient strip foldings are efficient approximation
algorithms for \emph{milling} a grid polyhedron.
Motivated by rapid-fabrication CNC milling/cutting tools,
milling problems are typically stated in terms of a 2D region called
a ``pocket'' and a cutting tool called a ``cutter'', with the goal being to
find a path or tour for the cutter that covers the entire pocket. 
In our situation, the ``pocket'' is the surface of the grid polyhedron,
and the ``cutter'' is a unit square constrained to move from one grid square
of the surface to an (intrinsically) adjacent grid square.

The typical goals in milling problems are to minimize the length of the tour
\cite{Arkin2000}
or to minimize the number of turns in the tour \cite{Arkin2005}.
Both versions are known to be strongly NP-hard, even when the pocket is
an integral orthogonal polygon and the cutter is a unit square.
We conjecture that the minimum-turn problem remains strongly NP-hard
when the pocket is a grid polyhedron, but the minimum-length problem
turns out to be polynomial: the dual graph is planar and
4-connected and thus Hamiltonian \cite{Bodini-Lefranc-2006},
so there is an ideal tour visiting each square exactly once,
which can be found in linear time \cite{Chiba-Nishizeki-1989}.
(This result also implies a 2-approximation for the metric of length
plus turns.)

In our situation, both length and number of turns are important,
as both influence the required length of a strip to cover the surface.
Thus we develop one algorithm that simultaneously approximates both measures.
Such results have also been achieved for 2D pockets \cite{Arkin2005};
our results are the first we know for surfaces in 3D.
%
%
Specifically,
we develop in Section~\ref{StripFolding:sec:Milling} an approximation algorithm
for computing a milling tour of a given grid polyhedron, achieving both a
$2$-approximation in length and an $8/3$-approximation in number of turns.

\paragraph{Fonts.}
To illustrate the power of strip folding, we designed a typeface, representing
each letter of the alphabet by a folding of a $1 \times x$ strip for some~$x$,
as shown in Figure~\ref{font}.
The individual-letters typeface consists of two fonts:
the unfolded font is a puzzle to figure out each letter,
while the folded font is easy to read.
These crease patterns adhere to an integer grid with orthogonal and/or
diagonal creases, but are not necessarily subpatterns of the canonical
hinge pattern.  This extra flexibility gives us control to produce folded
half-squares as desired, increasing the font's fidelity.

We have developed a web app that visualizes the font,%
\footnote{\url{http://erikdemaine.org/fonts/strip/}}
and chains letters together into one long strip folding.
Figure~\ref{title}, and
Figure~\ref{the end} at the end of the paper, show some examples.

\begin{figure}
  \centering
  \includegraphics[width=\linewidth]{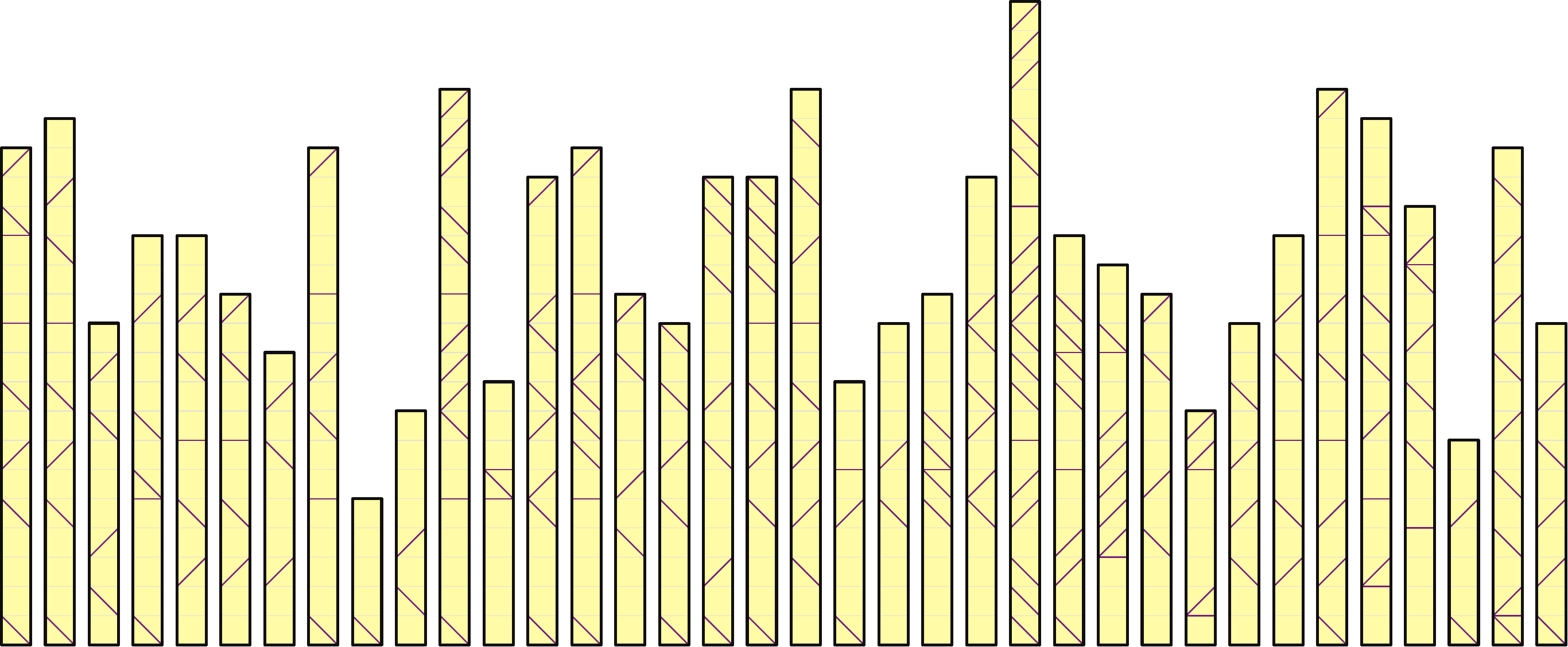}

  \bigskip

  \includegraphics[width=\linewidth]{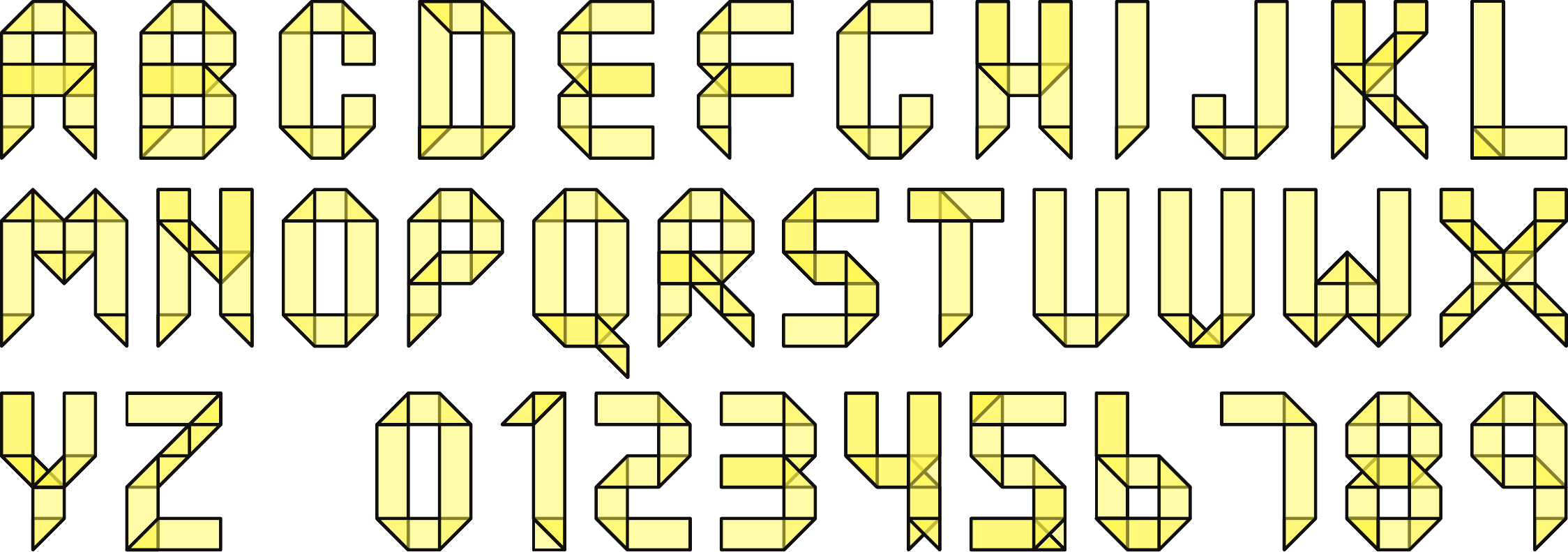}

  \caption{Strip folding of individual letters typeface, A--Z and 0--9:
    unfolded font (top) and folded font (bottom), where
    the face incident to the bottom edge remains face-up.}
  \label{font}
\end{figure}

\begin{figure}
\centering
\includegraphics[width=\linewidth]{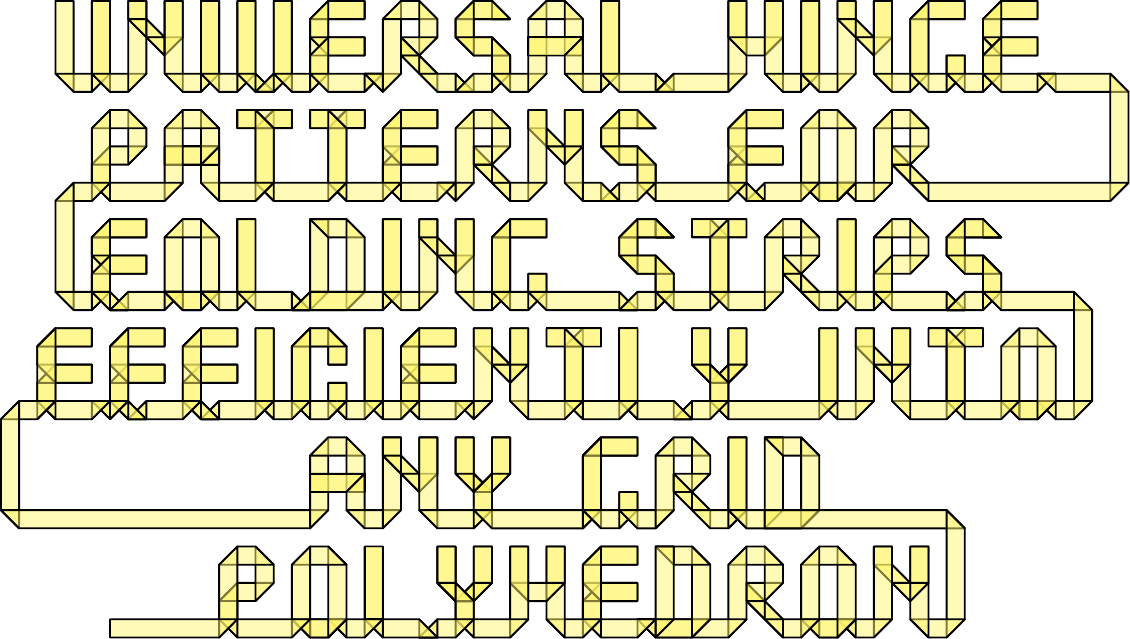}
\def\link{\url{http://erikdemaine.org/fonts/strip/}}
\caption{The paper title rendered with our web app, \link.}
\label{title}
\end{figure}

\section{Universality}
\label{Universality}

In this section, we prove that both the canonical strip and zig-zag strip of
Figure~\ref{StripFolding:fig:StripTypes}, of sufficient length,
can fold into any grid surface.  We begin with milling tours which provide
an abstract plan for routing the strip, and then turn to the details of how
to manipulate each type of strip.

\paragraph{Dual graph.}
Recall that a \emph{grid surface} consists of one or more
\emph{grid squares}---that is, squares of the 3D cube grid---glued
edge-to-edge to form a connected surface (ignoring vertex connections).
Define the \emph{dual graph} to have a dual vertex for each such grid square,
and a dual edge between the two dual vertices corresponding to any two grid
squares sharing an edge.
Our assumption of the grid surface being connected
is equivalent to the dual graph being connected.

\paragraph{Milling tours.}
A \emph{milling tour} is a (not necessarily simple) spanning cycle
in the dual graph, that is, a cycle that visits every dual vertex
at least once (but possibly more than once).
Equivalently, we can think of a milling tour as the path traced by
the center of a moving square that must cover the entire surface
while remaining on the surface, and return to its starting point.
Milling tours always exist: for example, we can double a spanning tree
of the dual graph to obtain a milling tour of length less than double the
given surface area.

At each grid square, we can characterize a milling tour as going straight,
turning, or U-turning---intrinsically on the surface---according to
which two sides of the grid square the tour enters and exits.
If the sides are opposite, the tour is \emph{straight};
if the sides are incident, the tour \emph{turns}; and if the sides
are the same, the tour \emph{U-turns}.
Intuitively, we can imagine unfolding the surface and developing the tour
into the plane, and measuring the resulting planar turn angle at the center
of the grid square.

\paragraph{Strip folding.}
To prove universality, it suffices to show that a canonical strip or zig-zag
strip can follow any milling tour and thus make any grid polyhedron.
In particular, it suffices to show how the strip can go straight, turn left,
turn right, and U-turn.
Then, in 3D, the strip would be further folded at each traversed edge of the
grid surface, to stay on the surface.  
Indeed, U-turns can be viewed as folding onto the opposite side of the same
surface, and thus are intrinsically equivalent to going straight;
hence we can focus on going straight and making left/right turns.
%

\paragraph{Canonical strip.}

Figure~\ref{StripFolding:fig:CanonicalTurns} shows how a canonical strip can turn
left or right; it goes straight without any folding.
Each turn adds $1$ to the length of the strip,
and adds $2$ layers to part of the grid square where the turn is made. 
Therefore a milling tour of length $L$ with $t$ turns of a grid
surface can be followed by a canonical strip of length $L+t$.
Furthermore, if the milling tour visits each grid square at most $c$ times,
then the strip folding has at most $3 c$ layers covering any point of the
surface.

\begin{figure}
\centering
  \subfigure[\label{StripFolding:fig:CanonicalLeftTurn}]
            {\includegraphics[scale=0.34]{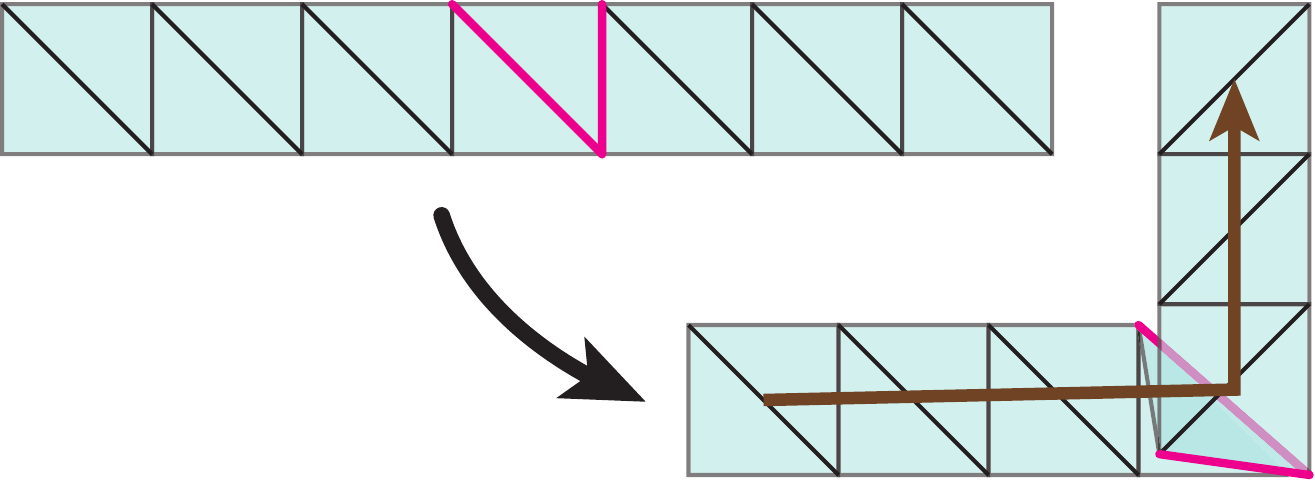}}\hfil\hfil\hfil\hfil
  \subfigure[\label{StripFolding:fig:CanonicalRightTurn}]
            {\includegraphics[scale=0.34]{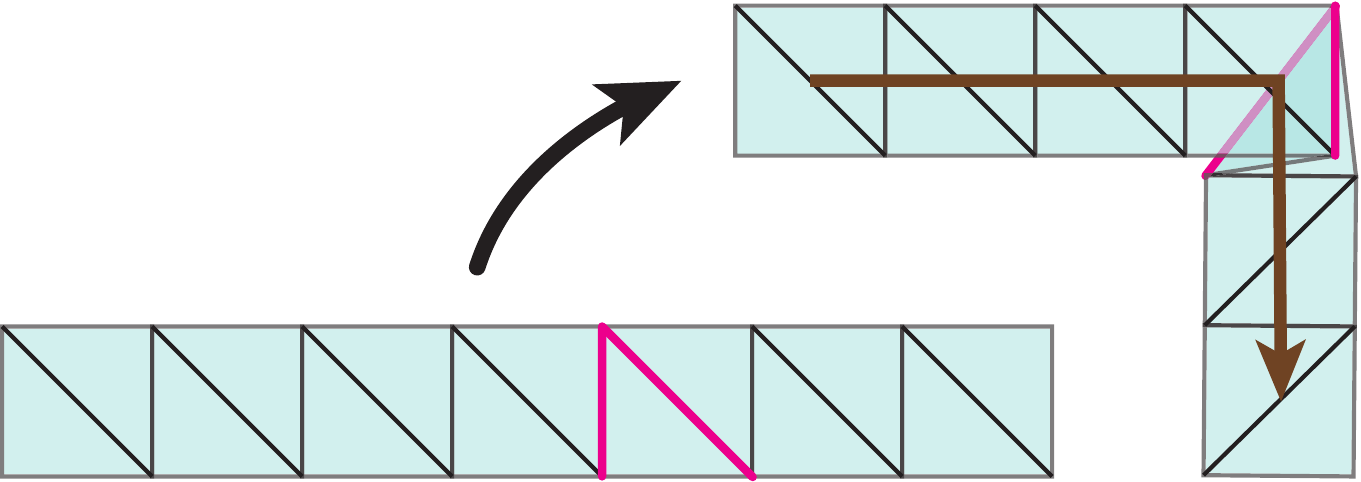}}
\caption{(a) Left and (b) right turn with a canonical strip.}
\label{StripFolding:fig:CanonicalTurns}
\end{figure}

\paragraph{Zig-zag strip.}

Figure~\ref{StripFolding:fig:ZigZagStraight} shows how to fold
a zig-zag strip in order to go straight.
In this straight portion, each square of the surface is covered
by two squares of the strip.
Figure~\ref{StripFolding:fig:ZigZagTurns} shows left and right turns.
Observe that turns require either one or three squares of the strip.
Therefore a milling tour of length $L$ with $t$ turns
can be followed by a zig-zag strip of length at most $2L+t$. 
Furthermore, if the milling tour visits each grid square at most $c$ times,
then the strip folding has at most $3 c$ layers
covering any point of the surface. 

\begin{figure}
\centering
\includegraphics[scale=0.4]{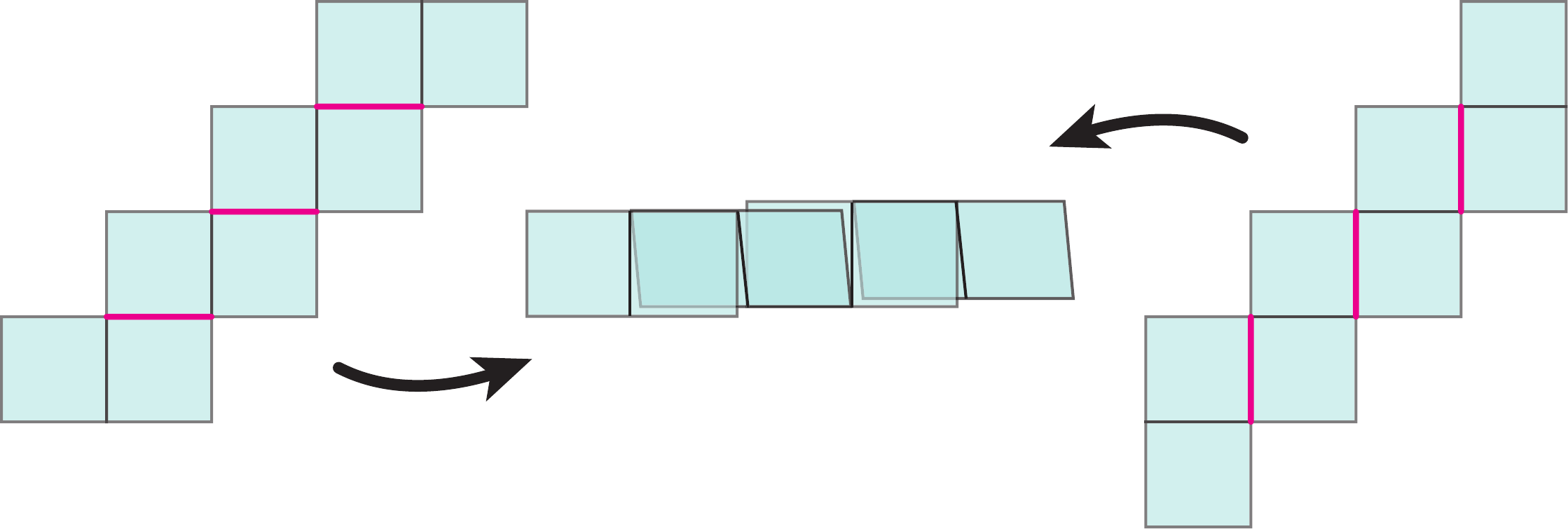}
\caption{Going straight with a zig-zag strip requires at most two unit squares per grid square.  Left and right crease patterns show two possible different parities along the strip.}
\label{StripFolding:fig:ZigZagStraight}
\end{figure}

\begin{figure}
\centering
  \subfigure[\label{StripFolding:fig:ZigZagOddTurns}]
            {\includegraphics[scale=0.3]{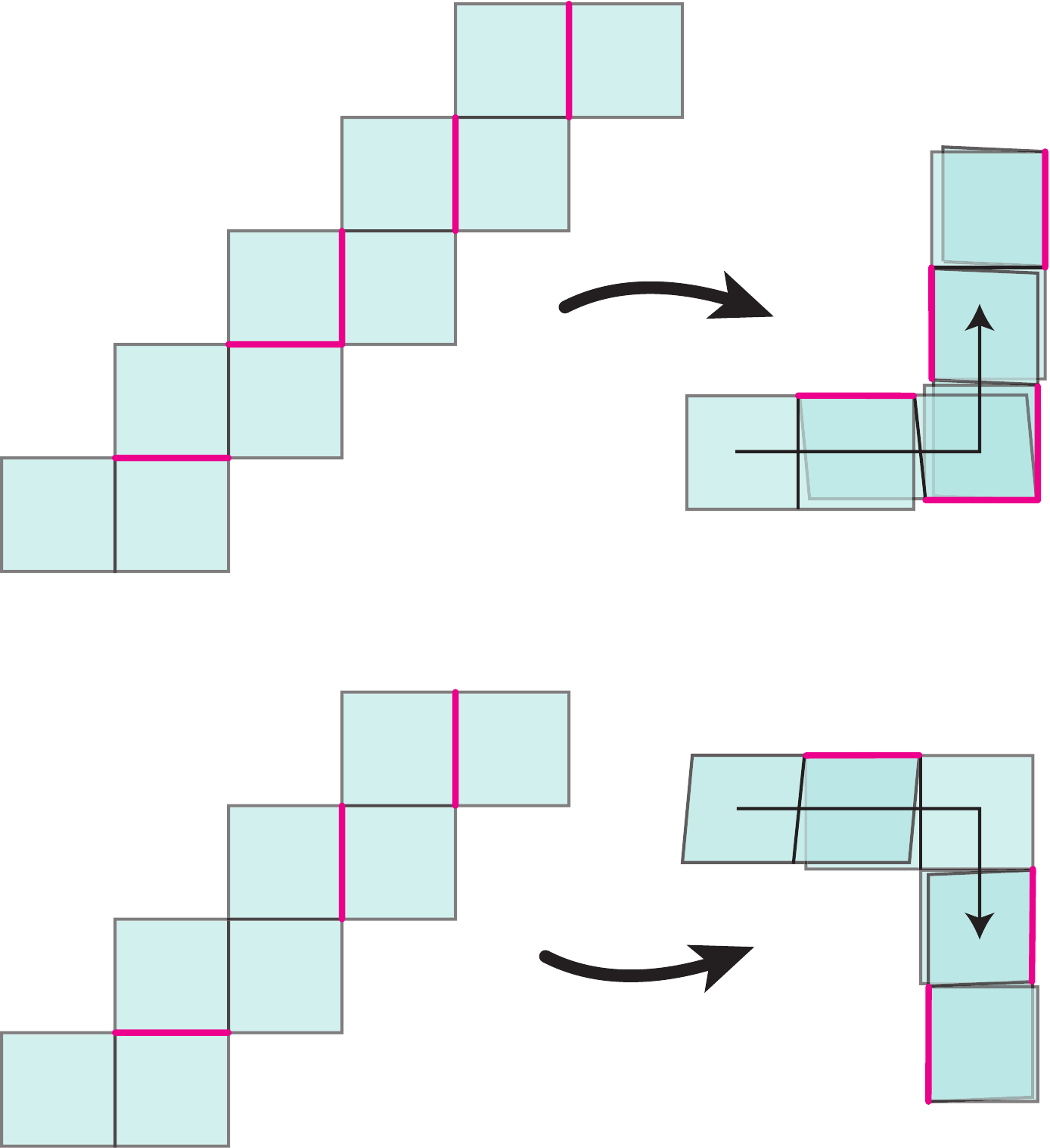}}\hfil\hfil\hfil\hfil
  \subfigure[\label{StripFolding:fig:ZigZagEvenTurns}]
            {\includegraphics[scale=0.3]{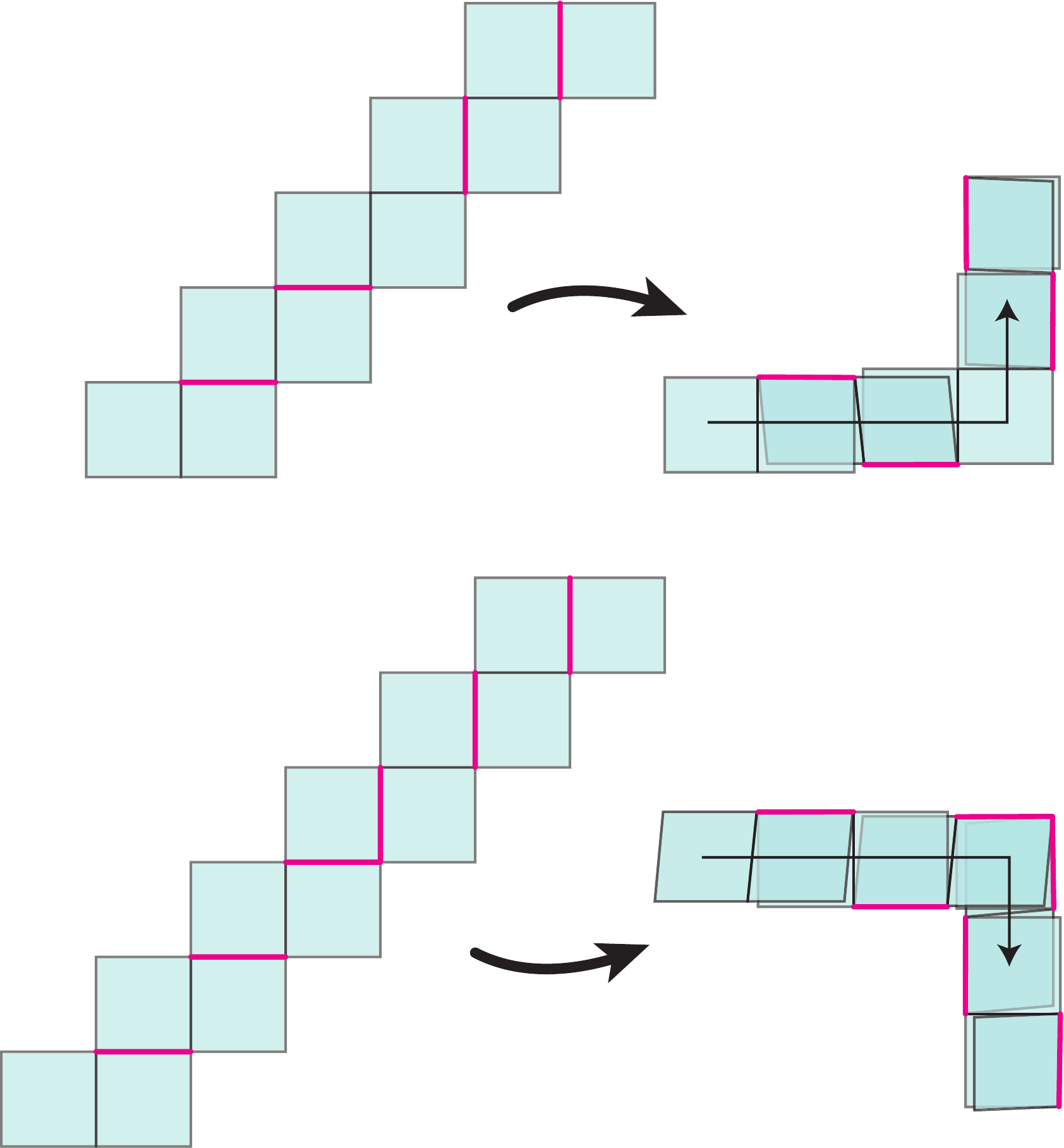}}

\caption{Turning with a zig-zag strip has two cases because of parity.
(a)~Turning left at an odd position requires three grid squares, whereas turning right requires one grid square.
(b)~Turning left at an even position requires one grid square, whereas turning right requires three grid squares.
}
\label{StripFolding:fig:ZigZagTurns}
\end{figure}

\begin{prop}
  Every grid surface of area $N$ can be folded from a
  canonical strip of length $4 N$, with at most eight layers stacked anywhere,
  and from a zig-zag strip of length $6 N$, with at most twelve layers
  stacked anywhere.
\end{prop}

\begin{proof}
  The doubled-spanning-tree milling tour has length $L < 2 N$,
  and in the worst case turns at every square visited, i.e., $t < 2 N$
  (counting U-turns at the leaves of the spanning tree).
  Each square gets visited by the tour at most four times (once per side).
  The bounds then follow from the analysis above.
\end{proof}

The goal in the rest of this paper is to achieve better bounds
for grid polyhedra, using more carefully chosen milling tours.
For comparison, the best approximation algorithm for milling tours in
polygons (not grid surfaces) \cite{Arkin2000} achieves a bound of
$L \leq 1.325 \, N$ which, combined with the trivial $t \leq L$ bound,
achieves a canonical-strip bound of $t+L \leq 2.65 \, N$.
We will beat this bound for grid surfaces, reducing to $2 N$.

\section{Milling Tour Approximation}
\label{StripFolding:sec:Milling}

This section presents a constant-factor approximation algorithm for
milling a (genus-0) grid polyhedron $P$ with respect to both length and turns.
Specifically, our algorithm is a $2$-approximation in length and an
$8/3$-approximation in turns.  Our milling tours also have special properties
that make them more amenable to strip folding.

Our approach is to reduce the milling problem to vertex cover in a tripartite
graph.  Then it follows that our algorithm is a $2 \alpha$-approximation
in turns, where $\alpha$ is the best approximation factor for vertex cover
in tripartite graphs. 
The best known bounds on $\alpha$ are $34/33 \leq \alpha \leq 4/3$.
Clementi et al.~\cite{Clementi1999} proved that minimum vertex cover in
tripartite graphs is not approximable within a factor smaller than
$34/33=1.\overline{03}$ unless $\text{P}=\text{NP}$.
Theorem~1 of \cite{Hochbaum1983} implies a $4/3$-approximation for minimum
weighted vertex cover for tripartite graphs (assuming we are given the
3-partition of the vertex set, which we know in our case).
Thus we use $\alpha=4/3$ below.
An improved approximation ratio $\alpha$ would improve our approximation
ratios, but may also affect the stated running times, which currently
assume use of \cite{Hochbaum1983}.

\subsection{Bands}
\label{Bands}

\begin{figure}
  \centering
  \begin{tabularx}{\linewidth}{XX}
    \vbox{\hbox{\subfigure[\label{dog1}A grid polyhedron from \cite{Moteins_TRO}]
      {\includegraphics[scale=0.47,page=1]{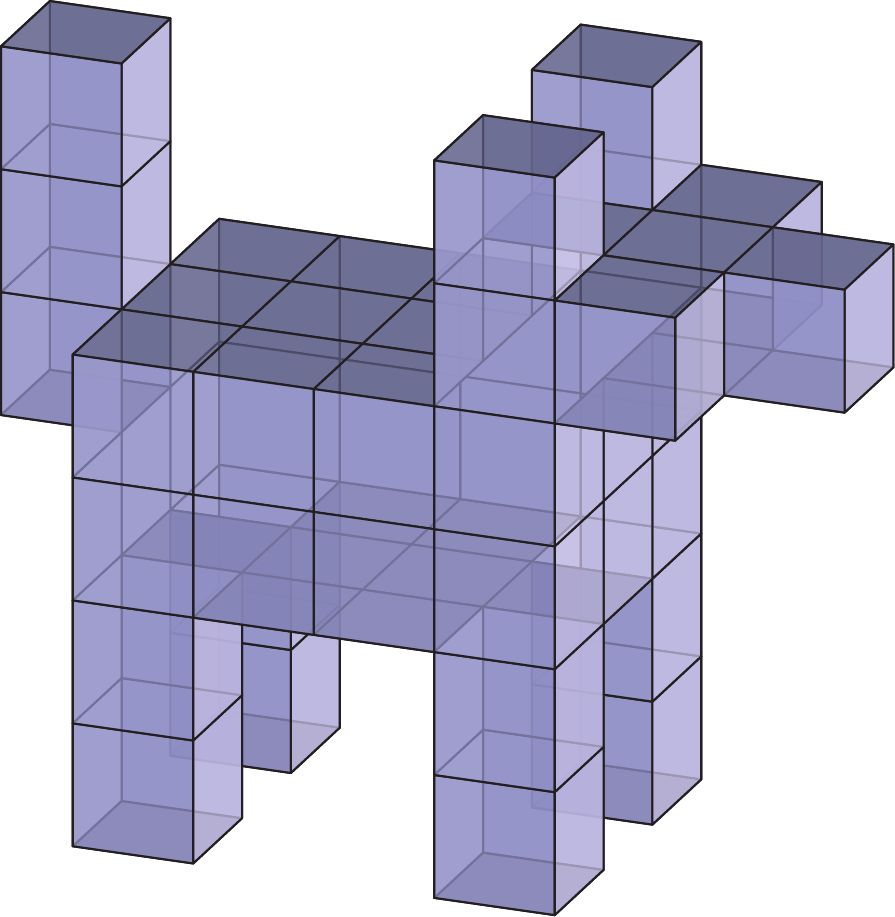}}}}&
    \vbox{\hbox{\subfigure[\label{dog2}$x$-bands]
      {\includegraphics[scale=0.47,page=3]{Figures/dog}}}}\\
    \vbox{\hbox{\subfigure[\label{dog3}$y$-bands]
      {\includegraphics[scale=0.47,page=2]{Figures/dog}}}}&
    \vbox{\hbox{\subfigure[\label{dog4}$z$-bands]
      {\includegraphics[scale=0.47,page=4]{Figures/dog}}}}
  \end{tabularx}
  \caption{Illustration of bands.}
  \label{bands}
\end{figure}

The basis for our approximation algorithms is
the notion of ``bands'' for a grid polyhedron~$P$;
refer to Figure~\ref{bands}.
Let $x_\min$ and $x_\max$ respectively be the minimum and maximum $x$
coordinates of $P$; define $y_\min, y_\max,$ $z_\min, z_\max$ analogously.
These minima and maxima have integer values because the vertices of $P$ lie on
the integer grid.
Define the $i$th \emph{$x$-slab} $S_x(i)$ to be the slab bounded by
parallel planes $x=x_\min+i$ and $x=x_\min+i+1$,
for each $i\in\{0,1,\dots, x_\max-x_\min -1\}$.
The intersection of $P$ with the $i$th $x$-slab $S_x(i)$
(assuming $i$ is in the specified range) is either a single \emph{band}
(i.e., a simple cycle of grid squares in that slab),
or a collection of such bands, which we refer to as $x$-bands.
Define $y$-bands and $z$-bands analogously.

Two bands \emph{overlap} if there is a grid square contained in both bands. 
Each grid square of $P$ is contained in precisely two bands
(e.g., if a grid square's outward normal were in the $+z$-direction,
then it would be contained in one $x$-band and one $y$-band).
Two bands $B_1$ and $B_2$ are \emph{adjacent} if they do not overlap,
and a grid square of $B_1$ shares an edge with a grid square of~$B_2$.
A \emph{band cover} for $P$ is a collection of $x$-, $y$-, and $z$-bands
that collectively cover the entire surface of~$P$.
The \emph{size} of a band cover is the number of its bands.

\subsection{Cover Bands}
\label{StripFolding:Covering Bands}

The starting point for the milling approximation algorithm is to find an
approximately minimum band cover, as the minimum band cover is a lower bound
on the number of turns in any milling tour:

\begin{prop}\label{StripFolding:prop:BCLowerBoundOnTurns}
{\rm \cite[Lemma~4.9]{Arkin2005}}
The size of a minimum band cover of a grid polyhedron $P$
is a lower bound on the number of turns in any milling tour of~$P$.
\end{prop}
\begin{proof}
Consider a milling tour of $P$ with $t$ turns.
Extend the edges of the tour of $P$ into bands.
This yields a band cover of size at most~$t$.
\end{proof}


Next we describe how to find a near-optimal band cover.
Consider the graph $G_P$ with one vertex per band of a grid polyhedron~$P$,
connecting two vertices by an edge if their corresponding bands overlap.
It turns out that an (approximately minimum) vertex cover in $G_P$
will give us an (approximately minimum) band cover in~$P$:

\begin{prop}\label{StripFolding:prop:VertexCoverBandCover}
A vertex cover for $G_P$ induces a band cover of the same size and vice versa.
\end{prop}

\begin{proof}
Because the vertices of $G_P$ correspond to bands, it suffices to show that a set of bands $S$ covers all
the grid squares of $P$ if and only if the corresponding vertex set $V_S$ covers all the edges of $G_P$.
First suppose that $V_S$ covers the edges of $G_P$. 
Observe that a grid square is contained in two bands and those bands overlap, so there is a corresponding
edge in $G_P$.  Because $V_S$ covers the edge in $G_P$, $S$ must cover the grid square.

Conversely, suppose $S$ covers all the grid squares of $P$.  Any edge of $G_P$ corresponds to two
overlapping bands, which overlap in a grid square (in fact, they overlap in at least two grid squares because $P$ has genus $0$).
Because the grid square is in exactly two bands, and $S$ covers the grid square, then $V_S$ covers
the corresponding edge in $G_P$.  
%
\end{proof}

Because the bands fall into three classes ($x$-, $y$-, and $z$-),
with no overlapping bands within a single class,
$G_P$ is tripartite.
Hence we can use an $\alpha$-approximation algorithm for vertex cover
in tripartite graphs to find an $\alpha$-approximate vertex cover in $G_P$
and thus an $\alpha$-approximate band cover of~$P$.


\subsection{Connected Bands}

Our next goal will be to efficiently tour the bands in the cover.
Given a band cover $S$ for a grid polyhedron $P$, define the \emph{band graph} $G_S$ to be
the subgraph of $G_P$ induced by the subset of vertices corresponding to~$S$. 
We will construct a tour of the bands $S$ based on a spanning tree of $G_S$. 
Our first step is thus to show that $G_S$ is connected
(Lemma~\ref{StripFolding:lem:connectivity} below).
We do so by showing that adjacent bands
(as defined in Section~\ref{Bands})
are in the same connected component of~$G_S$,
using the following lemma of Genc \cite{Genc2008}:

\begin{lemma} \label{genc lemma}
  {\rm \cite{Genc2008}}%
  \footnote{Genc \cite{Genc2008} uses somewhat different terminology
    to state this lemma: ``straight cycles in the dual graph'' are our bands,
    and ``crossing'' is our overlapping.  The induced subgraph is also defined
    directly, instead of as an induced subgraph of~$G_P$.}
  For any band $B$ in a grid polyhedron~$P$,
  let $N_B$ be the bands of $P$ overlapping~$B$.
  (Equivalently, $N_B$ is the set of neighbors of $B$ in~$G_P$).
  Then the subgraph of $G_P$ induced by $N_B$ is connected.
\end{lemma}

\begin{lemma}\label{StripFolding:lem:connectivity}
If $S$ is a band cover for a grid polyhedron $P$, then the graph $G_{S}$ is connected.
\end{lemma}
\begin{proof}
First define an auxiliary graph $H_S$ which has edges between adjacent bands in addition to edges between overlapping bands.  
Because the surface of $P$ is connected, $H_S$ is connected, because there is
an (orthogonal) path on the surface of $P$ between any two bands and this path
induces a sequence of edges in~$H_S$.  
We will argue that $G_S$ has the same connectivity as $H_S$,
and thus is connected.

Now consider any pair of adjacent bands $B_1$ and $B_2$ in~$S$,
forming an edge in $H_S$ but not~$G_S$.
Refer to Figure~\ref{B123}.
Because $B_1$ and $B_2$ are adjacent, there exists grid squares $g_1\in B_1$ and $g_2 \in B_2$ which share an edge $e$.
Let $B_3$ be the band through grid squares $g_1$ and $g_2$.  
If $B_3$ is in the band cover,
then the vertices corresponding to $B_1$ and $B_2$ are in the same connected component of $G_S$ because $B_3$ overlaps both $B_1$ and~$B_2$.
Thus suppose that $B_3$ is not in~$S$.  
Then each of the bands overlapping $B_3$ (excluding $B_3$ itself) must be in $S$; otherwise there would be an uncovered grid square of~$B_3$.  
Let $N_3$ denote this set of bands.
By Lemma~\ref{genc lemma}, the subgraph $G_P[N_3]$ of $G_P$ induced by $N_3$
is connected.  Now because $G_P[N_3]$ is also a subgraph of~$G_S$, 
we have that the vertices corresponding to $B_1$ and $B_2$ are in the same connected component of~$G_S$.
\end{proof}

\begin{figure}
  \centering
  \includegraphics[scale=0.75]{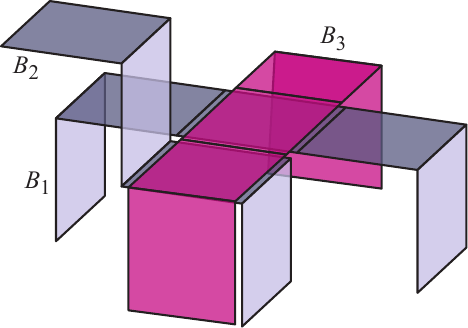}
  \caption{Bands $B_1$, $B_2$, and $B_3$ in the proof of
    Lemma~\ref{StripFolding:lem:connectivity}.}
  \label{B123}
\end{figure}


\subsection{Band Tour}

Now we can present our algorithm for transforming a band
cover into an efficient milling tour.

\begin{theorem}
\label{StripFolding:thm:MillingPseudopoly}
Let $P$ be a grid polyhedron with $N$ grid squares.
In $O(N^2 \log N)$ time, we can find a milling tour of $P$ that is
a $2$-approximation in length and an $8/3$-approximation
(or more generally, a $2 \alpha$-approximation) in turns.
\end{theorem}

\begin{proof}
We start with an $\alpha$-approximation to the minimum vertex cover of $G_P$,
such as the $\alpha = 4/3$ algorithm of \cite{Hochbaum1983}.
Then we convert this vertex cover into a band cover~$S$. 
By Proposition~\ref{StripFolding:prop:VertexCoverBandCover}, $S$ is an $\alpha$-approximation
for the minimum band cover for~$P$.
Construct the band graph $G_S$, and construct a spanning tree $T$ of $G_S$,
which exists as $G_S$ is connected by Lemma~\ref{StripFolding:lem:connectivity}. 
We now describe how to obtain a milling tour by touring~$T$.
(Note that we will use ``nodes of the tree'' and ``bands of $S$'' interchangeably.)

Consider the band $r$ which is the root of $T$.
Pick an arbitrary starting grid square $p_r$ on~$r$, as well as a starting direction along band~$r$.
Now order the children of $r$ based on the order in which they are first encountered as we walk along $r$ starting at $p_r$.
We will visit the children of $r$ in this order when we follow the tour.
Let the ordered children of $r$ be denoted by $u_1,\dots, u_k$. 
Walking along~$r$, when we hit a child $u_i$ (say at grid square $p_i$) which has not yet been visited,
we turn right (or left, it doesn't matter) onto band $u_i$ and recursively visit the children of $u_i$ as we are walking,
until we eventually return back to $r$ which requires an additional turn (in the opposite direction of the turn onto $u_i$) to get back onto $r$ and resume walking along~$r$. 
Because turns only ever occur between parent and child pairs, we can bound the number of turns in terms of the
number of edges of~$T$. In particular, there are two turns per edge of $T$, one turn from the parent band
onto the child band, and one turn from the child band back onto the parent band. 
Thus the total number of turns is
$$2\,|E(T)|= 2\cdot(|S|-1)\leq 2\,\alpha\cdot \text{OPT}_{BC},$$
where $\text{OPT}_{BC}$ is the size of the minimum band cover for $P$.   
Now by Proposition~\ref{StripFolding:prop:BCLowerBoundOnTurns}, $\text{OPT}_{BC}$ is a lower bound on the number
of turns necessary, so this yields a $2\alpha$-approximation in turns.  

Next we argue that each grid square is visited at most twice,
implying a $2$-approximation in length. 
Each grid square is contained in exactly two bands of $P$, at least one of which must be in $S$ because $S$ is a band cover.
Let $g$ be a grid square of $P$.  If $g$ is contained in just one band of $S$, then it is visited only once,
namely when that band is traversed by the milling tour.  
Thus suppose $g$ is contained in two bands, $S_1$ and $S_2$, of $S$.  There are two cases to consider. 
In the first case,  neither band is a parent of the other band.
So $g$ will be visited exactly two times, once when $S_1$ is traversed, and a second time when $S_2$ is traversed.
(We call such a $g$ a \emph{straight junction} due to the way in which it is traversed, as Figure~\ref{StripFolding:fig:StraightJunction}
depicts.)
The second case to consider is when one of the bands is a parent of the other band.  Without loss of generality, assume
$S_1$ is a parent of $S_2$.  Then $g$ will be visited twice,
once when turning from $S_1$ onto $S_2$, and a second time when turning from $S_2$ back onto $S_1$. 
(We call $g$ a \emph{turn junction} in this case, as Figure~\ref{StripFolding:fig:TurnJunction} depicts.
Notice that the turn from $S_1$ onto $S_2$ is in the opposite direction of the turn from $S_2$ onto $S_1$.
This fact will be useful when we explore applications of this algorithm to strip folding, although
it is not immediately of use.)
Hence each grid square is visited at most twice.


\begin{figure}
\centering
  \subfigure[\label{StripFolding:fig:StraightJunction}]
            {\includegraphics[scale=0.7]{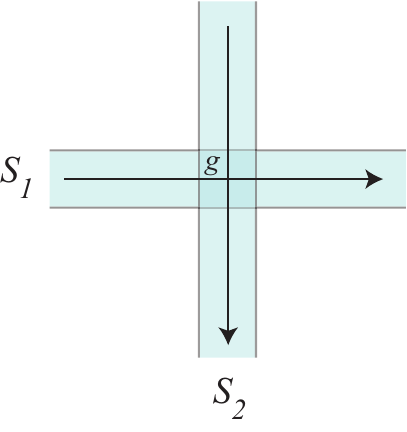}}\hfil\hfil
  \subfigure[\label{StripFolding:fig:TurnJunction}]
            {\includegraphics[scale=0.7]{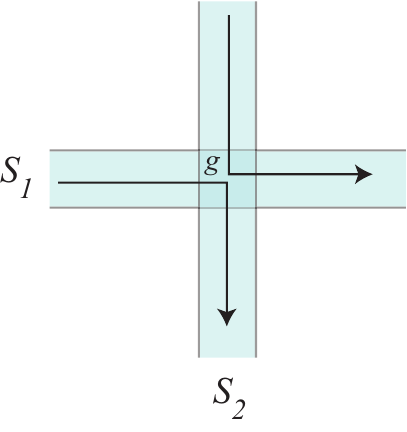}}
\caption{(a) A straight junction, and (b) a turn junction.}
\label{StripFolding:fig:Junctions}
\end{figure}

Finally we analyze the running time of this algorithm.
The time to compute an $\alpha$-approximation for the minimum
band cover of $P$ is the same as the time to compute an
$\alpha$-approximation for the minimum vertex cover for $G_P$.
Because each grid square corresponds to an edge of $G_P$ and each edge of $G_P$ corresponds to two unique grid squares, 
the number of edges $|E(G_P)|=N/2$.
The number of vertices $|V(G_P)|$ is the number of unit bands of $P$,
which is at most~$N$.
By Theorem~1 of \cite{Hochbaum1983}, the time to compute a $4/3$-approximation
for the minimum vertex cover for $G_P$ is
$O(|V(G_P)| |E(G_P)| \log\,|V(G_P)|)$, which is thus $O(N^2\log\,N)$.
This term is dominant in the running time for computing the milling tour,
establishing the theorem. 
%
\end{proof}

We now state some additional properties of any milling tour produced by the approximation algorithm of Theorem~\ref{StripFolding:thm:MillingPseudopoly},
which will be useful for later applications to strip folding in Section~\ref{StripFolding:sec:StripFolding}.
\begin{prop}
\label{StripFolding:prop:MillingTourProperties}
Let $P$ be a grid polyhedron, and consider a milling tour of $P$ obtained from the approximation algorithm
of Theorem~\ref{StripFolding:thm:MillingPseudopoly}.  Then the following properties hold:
\begin{enumerate}
\item A grid square of $P$ is either visited once, in which case it is visited by a straight part of the tour; or it is visited twice,
in one of the two configurations of Figure~\ref{StripFolding:fig:Junctions} (a straight junction or a turn junction---in particular, never a U-turn).
\item In the case of a turn junction, the length of the milling tour between the two visits to the grid square (counting
only one of the two visits to the grid square in the length measurement) is even.
\item  The tour can be modified to alternate between left and right turns (without changing its length or the number of turns).
\end{enumerate}
\end{prop}
\begin{proof}
Property (1) has already been established in the proof of Theorem~\ref{StripFolding:thm:MillingPseudopoly} (see the
paragraph on the length calculation which bounds the number of times each grid square is visited and how
a grid square can be visited). 

The proof of Property (2) follows by induction on the number of bands in the spanning tree from which the milling tour was constructed. 
The base case of one band is trivial as a band has an even number of grid squares.  Thus suppose Property (2) holds
for all such trees with fewer than $k$ bands for some $k > 1$.  Assume we have a tree $T$ with $k$ bands. 
Consider a leaf of $T$.
Let $B_2$ be the band corresponding to the leaf, and let $B_1$ be the band corresponding to the parent of $B_2$. 
If we remove $B_2$ from $T$ we get a tree $T'$ with $k-1$ bands, so by induction Property (2) holds. 
Adding $B_2$ back, the length of the milling tour between the two visits to a grid square will either
be unaffected (in the case where the subtree rooted at the child band corresponding to the grid square does not contain $B_2$)
 or it will increase by $|B_2|$ (i.e., the number of grid squares in $B_2$) which is even.
Thus Property (2) still holds.  

Property (3) follows by the way we constructed the milling tour as a tour of a spanning tree of bands.  
We had the freedom to choose whether to turn left or right from a parent band onto a child band,
and then the turn direction from the child back onto the parent was forced to be the opposite turn direction. 
Hence we can choose the turn directions from parent bands onto child bands in such a way that maintains alternation
between left and right turns.  
\end{proof}

\subsection{Polynomial Time}
\label{StripFolding:sec:Poly}

The algorithm described above is polynomial in the surface area $N$ of the
grid polyhedron, or equivalently, polynomial in the number of unit cubes
making up the polyomino solid.  For our application to strip folding,
this is polynomial in the length of the strip, and thus sufficient for
most purposes.  On the other hand, polyhedra are often encoded as a collection
of $n$ vertices and faces, with faces possibly much larger than a unit square.
In this case, the algorithm runs in \emph{pseudopolynomial} time.

Although we do not detail the approach here, our algorithm can be modified to
run in polynomial time.  To achieve this result,
we can no longer afford to deal with unit bands directly, because their
number is polynomially related to the number $N$ of grid squares,
not the number $n$ of vertices.
To achieve polynomial time in $n$, we concisely encode the output milling tour
using ``fat'' bands rather than unit bands,
which can then be easily decoded into a tour of unit bands.
By making each band as wide as possible, their number is polynomially
related to $n$ instead of~$N$.
Details of an $O(n^3 \log n)$-time milling approximation algorithm
(with the same approximation bounds as above)
can be found in \cite{nadia-thesis}.

\section{Strip Foldings of Grid Polyhedra} \label{StripFolding:sec:StripFolding}
In this section, we show how we can use the milling tours from
Section~\ref{StripFolding:sec:Milling}
to fold a canonical strip or zig-zag strip
efficiently into a given (genus-0) grid polyhedron.
For both strip types, define the \emph{length} of a strip to be the number of
grid squares it contains; refer to Figure~\ref{StripFolding:fig:StripTypes}.
For a strip folding of a polyhedron~$P$, define the \emph{number of layers covering a point} $q$ on $P$ to be
the number of interior points of the strip that map to $q$ in the folding,
excluding \emph{crease points}, that is, points lying on a hinge that gets
folded by a nonzero angle. 
(This definition may undercount the number of layers along one-dimensional
 creases, but counts correctly at the remaining two-dimensional subset of~$P$.)
We will give bounds on the length of the strip and also on the maximum number of layers of the strip covering
any point of the polyhedron.  

%

\subsection{Canonical Strips}

The main idea for canonical strips is that Properties (1) and (3)
of Proposition~\ref{StripFolding:prop:MillingTourProperties}
allow us to make turns
as shown in Figure~\ref{StripFolding:fig:turnJunctionCanonical},
so that we do not waste an extra square of the strip per turn. 
\begin{theorem}
\label{StripFolding:thm:CanonicalFolding}
If a grid polyhedron $P$ has a milling tour $\T$ of length $L$ with $t$ turns, and $\T$ satisfies
Properties (1) and (3) of Proposition~\ref{StripFolding:prop:MillingTourProperties},
then $P$ can be covered by a folding of a canonical strip of length $L$
with at most two layers covering any point of $P$.
Furthermore, the number of hinges folded by $180^\circ$ is $t$. 
\end{theorem}
\begin{proof}
Because $\T$ alternates between left and right turns (Property (3)), 
the canonical strip can follow $\T$ by making a single diagonal fold by
$180^\circ$ at each turn (instead of the more complicated folds given by
Figure~\ref{StripFolding:fig:CanonicalTurns}).
(It is the fact that the canonical strip has all the diagonal hinges in the same direction
that forces alternation between left and right turns from folding only at diagonal hinges.)
Furthermore, by Property (1), if $\T$ turns at a grid square, then the grid square is visited
 twice and must correspond to a turn junction as shown in Figure~\ref{StripFolding:fig:TurnJunction}.  
Thus we can cover turn junctions with the canonical strip as shown in Figure~\ref{StripFolding:fig:turnJunctionCanonical}.
Notice that one grid square of the strip is used per turn, and there are only two layers of the strip
covering a turn junction.  

By Property (1), if $\T$ goes straight at a grid square, then it is visited at most twice by $\T$ (and in 
the case of being visited twice, it is a straight junction as shown in Figure~\ref{StripFolding:fig:StraightJunction}).
Going straight along $\T$ uses just one grid square of the strip, so a straight junction is covered
by two layers of the strip.  The total length of the strip used is thus $L$, 
and there are at most two layers of the strip covering each point of $P$.  
The number of hinges folded by $180^\circ$ is exactly one turn in the milling tour, which is $t$; all folds made for the strip to follow the surface are by $90^\circ$ because, by Property~(1), the tour has no U-turns.
\end{proof}

\begin{corollary}
\label{StripFolding:cor:CanonicalFolding}
Let $P$ be a grid polyhedron, and let $N$ be the number of grid squares of $P$.  Then $P$ can be covered
by a folding of a canonical strip of length $2 N$, and with at most two layers covering any point of $P$.
\end{corollary}

\begin{figure}
\centering
\includegraphics[scale=0.6]{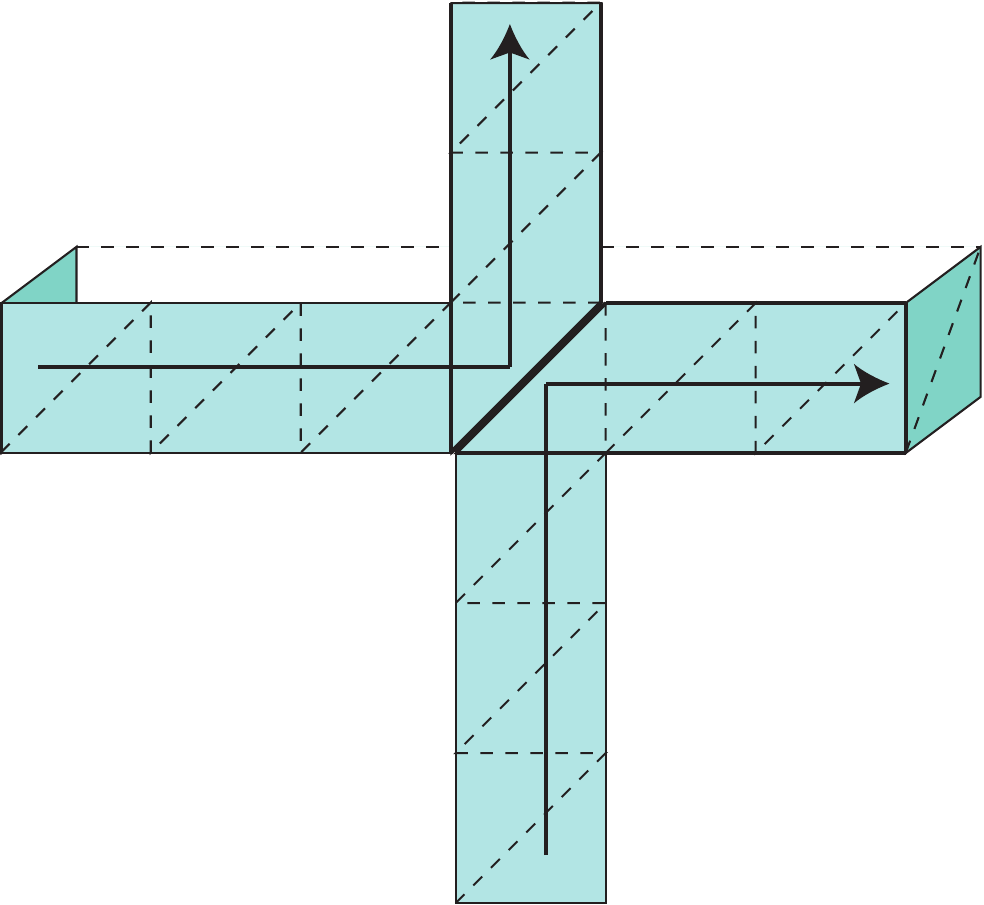}
\caption{A turn junction for a canonical strip.  The parent band is vertical and the child band is horizontal. 
The canonical strip starts out traveling upwards, turns right onto the child, visits the child and its subtree, and returns from the child
by turning left back onto the parent band. 
Notice that each turn uses just one grid square of the canonical strip, and there are only two layers of paper covering
the turn junction grid square. }
\label{StripFolding:fig:turnJunctionCanonical}
\end{figure}

\subsection{Zig-Zag Strips}

For zig-zag strips, we instead use Properties (1) and (2)
of Proposition~\ref{StripFolding:prop:MillingTourProperties}:
\begin{theorem}
\label{StripFolding:thm:ZigZagFolding}
If a grid polyhedron $P$ has a milling tour $\T$ of length $L$ with $t$ turns, and $\T$ satisfies
properties (1) and (2) of Proposition~\ref{StripFolding:prop:MillingTourProperties},
then $P$ can be covered by a folding of a zig-zag strip of length $2L$
with at most four layers covering any point of $P$.
Furthermore, the number of hinges folded by $180^\circ$ is $L$. 
\end{theorem}
\begin{proof}
By Property~(1), $\T$ has no U-turns.  If we follow the tour with a zig-zag
strip using the turn gadgets in Figure~\ref{StripFolding:fig:ZigZagTurns},
then, at even positions, left turns require one unit square of the strip and
right turns require three unit squares, while at odd positions, left turns
require three unit squares and right turns require one unit square.
Hence the positions along $\T$ alternate between left turns being ``easy''
or being ``hard'' for the zig-zag strip, and similarly for right turns.

Consider a grid square $g$ corresponding to a turn junction of $\T$. Then $g$ is visited twice by the tour,
with the first visit and second visit having opposite turn directions (i.e., if the first visit is a left turn, the
second visit will be a right turn, and vice versa).  Now by Property (2), the length of the milling tour
between the first and second visit to $g$ (counting $g$ itself only once) is even, and so the second
visit to $g$ has the same parity as the first visit to~$g$.  Now because the turns are in opposite directions,
one of the turns costs one grid square of the zig-zag strip, while the other turn costs three grid squares
of the zig-zag strip as explained above. Hence turn junctions are covered by four grid squares of the zig-zag strip.

By Property (1), the only other types of grid squares to consider are a straight junction, or a grid square which 
is visited exactly once by going straight.  Figure~\ref{StripFolding:fig:ZigZagStraight} illustrates that going straight
at a grid square requires at most two squares of the zig-zag strip.  And because a straight junction is visited twice,
at most four squares of the zig-zag strip cover a straight junction.  
Combining this with the turn junction coverage of four, we have at most four layers of the strip covering any point of $P$. 

Now we are ready to measure the total length of the zig-zag strip used.  The length of $\T$ spent going
straight is $L-t$; thus we use $\leq 2(L-t)$ grid squares of the strip for going straight.  
The length of $\T$ spent on turns is $t$, and half of the turns cost $1$ grid square while the other half
cost $3$ grid squares, so we spend $\leq t/2 + 3t/2 =2t$ grid squares of the strip for turns. 
Thus the total number of grid squares of the strip used is at most $2(L-t) + 2t = 2L$.

Finally we bound the number of $180^\circ$ folds of the strip.  At portions of the tour going straight,
we use one $180^\circ$ fold per unit of length, and thus $L-t$ such folds are dedicated to going straight.
For turns, half of the turns require two $180^\circ$ folds, whereas the other half require no $180^\circ$ folds.
Hence we obtain a total of $L-t + 2(t/2) = L$ $180^\circ$ folds. 
\end{proof}

\begin{corollary}
\label{StripFolding:cor:ZigZagFolding}
Let $P$ be a grid polyhedron, and let $N$ be the number of grid squares of $P$.  Then $P$ can be covered
by a folding of a zig-zag strip of length $4 N$, and with at most four layers covering any point of $P$.
\end{corollary}

By coloring the two sides of the zig-zag strip differently, we can also bicolor the surface of $P$ in any pattern we wish,
as long as each grid square is assigned a uniform color.
We do not prove this result formally here, but mention that
the bounds in length would become somewhat worse,
and the rigid motions presented in Section~\ref{StripFolding:subsec:Collision}
do not work in this setting. 


\section{Rigid Motion Avoiding Collision}
\label{StripFolding:subsec:Collision}

So far we have focused on just specifying a final folded state for the strip,
and ignored the issue of how to continuously move the strip into that folded
state.  In this section, we show how to achieve this using a \emph{rigid
folding motion}, that is, a continuous motion that keeps all polygonal
faces between hinges rigid/planar, bending only at the hinges,
and avoids collisions throughout the motion.
Rigid folding motions are important for applications to real-world programmable
matter made out of stiff material except at the hinges.
Our approach may still suffer from practical issues, as it requires a large
(temporary) accumulation of many layers in an accordion form.

We prove that, if the grid polyhedron $P$ has feature size at least~$2$,
then we can construct a rigid motion of the strip folding without collision.
(By \emph{feature size at least $2$}, we mean that every exterior voxel of $P$
is contained in some $2\times 2\times 2$ box with empty interior.  For example,
Figure~\ref{dog1} does not have this property because of the voxels
between adjacent legs and ears.)
If the grid polyhedron we wish to fold does not have this property,
then one solution is to scale the polyhedron by a factor of~$2$,
and then the results here apply.

\subsection{Approach}

\begin{figure}
\centering

  \subfigure[\label{StripFolding:fig:CanonicalAccordion}]
            {\includegraphics[scale=1]{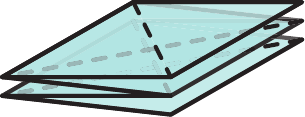}}
  \subfigure[\label{StripFolding:fig:ZigZagAccordion}]
            {\includegraphics[scale=1]{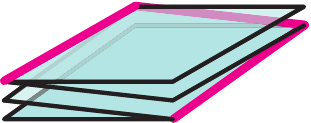}}
\caption{For both strip types, we can fold the unused portion of the strip
(initially, all of it) into an accordion to avoid collision during the folding
motion.  (a) Canonical strip. (b) Zig-zag strip, with hinges drawn thick for increased visibility.}
\label{StripFolding:fig:Accordion}
\end{figure}

Our approach is to keep the yet-unused portion of the strip folded up
into an \emph{accordion} and then to unfold only what is needed for the current
move: straight, left, or right.
Figure~\ref{StripFolding:fig:Accordion} shows the accordion state
for the canonical strip and for the zig-zag strip.
%
We will perform the strip folding in such a way that
the strip never penetrates the interior of the polyhedron~$P$,
and it never weaves under previous portions of the strip. 
Thus, we could wrap the strip around $P$'s surface
even if $P$'s interior were already a filled solid.
This restriction helps us think about folding the strip locally,
because some of $P$'s surface may have already been folded
(and it thus should not be penetrated) by earlier parts of the strip.

It suffices to show, regardless of the local geometry of the polyhedron
at the grid square where the milling tour either goes straight or turns,
and regardless of whether the accordion faces up or down relative to the
grid square it is covering (see Figure~\ref{StripFolding:fig:FaceUpDown}),
that we can maneuver the accordion in a way
that allows us to unroll as many squares as necessary to
perform the milling-tour move.
A na\"{\i}ve enumeration of the combinations of local geometry,
accordion facing up or down,
tour move (left turn, right turn, or straight),
and the type of strip being folded,
would lead to a large number of cases, but fortunately we can
narrow this space down to a small number of essential cases.

\begin{figure}
\centering
\subfigure[]{\includegraphics[width=0.45\linewidth]{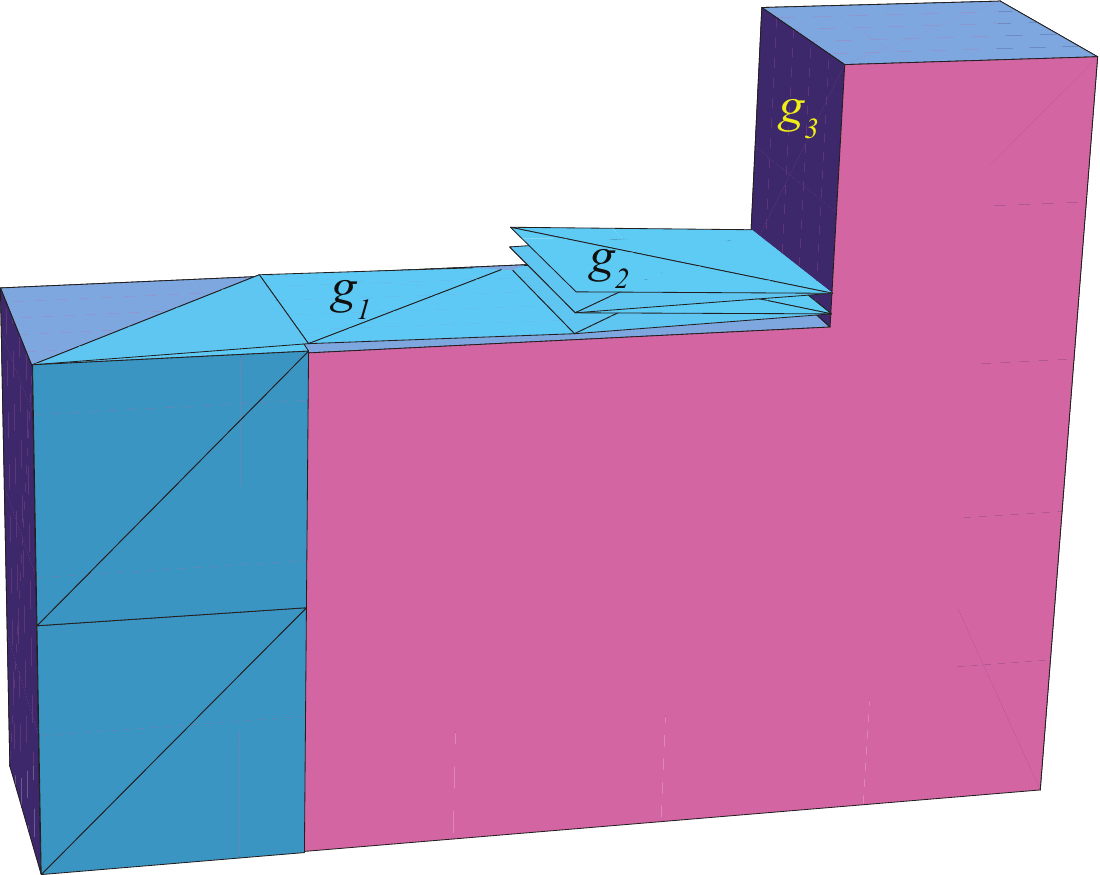}}
\hfill
\subfigure[]{\includegraphics[width=0.45\linewidth]{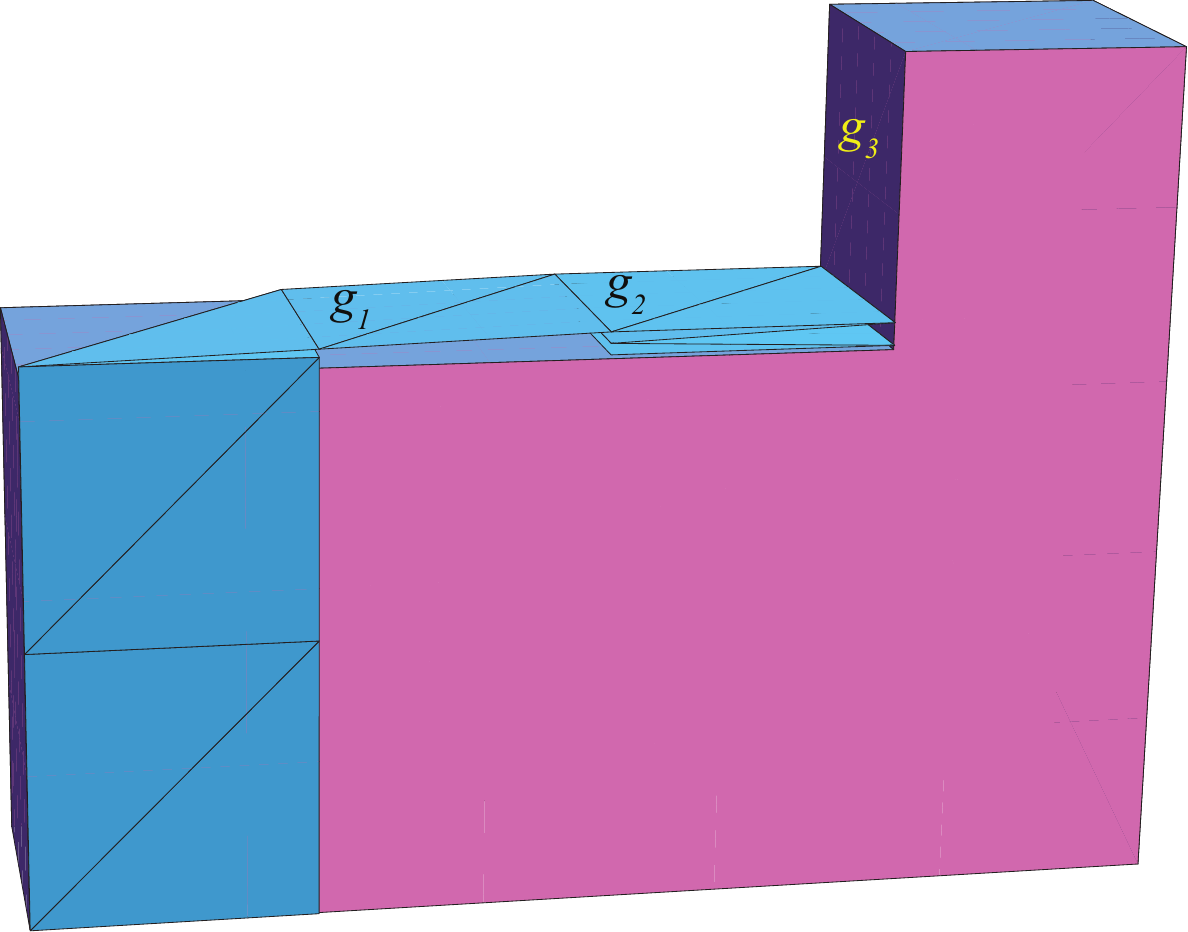}}
\caption{The accordion covering grid square $g_2$ (a) face up, and
  (b) face down.}
\label{StripFolding:fig:FaceUpDown}
\end{figure}

We use the following notation to describe such a local case.
Let $g_1$, $g_2$, and $g_3$ denote the grid square previously visited by the strip, the current grid square,
and the next grid square the strip will visit, respectively.  Let $e_{12}$ refer to the edge shared by $g_1$ and $g_2$,
and let $e_{23}$ refer to the edge shared by $g_2$ and $g_3$.
The local geometry of $P$ at $g_2$ is determined entirely by the
position of $e_{12}$ relative to $e_{23}$
(if they are adjacent edges of $g_2$, then we have a turn,
and if they opposite edges, then we have a straight)
and whether each of those edges is \emph{reflex}
($270^\circ$ dihedral angle), \emph{flat} ($180^\circ$ dihedral angle),
or \emph{convex} ($90^\circ$ dihedral angle), respectively.

Next we reduce the number of cases we need to consider.
First, left and right turns are symmetric, so we consider only right turns.
Second, when the accordion is facing up, going straight or turning is
relatively straightforward.  Figure~\ref{FaceUp} shows the moves for a
canonical strip going straight (an extrinsic rotation of some multiple of
$90^\circ$), and one case of turning (when $e_{23}$ is flat or convex);
we omit the remaining turning case (when $e_{23}$ is reflex) as it is
similar to the F--R turn case below.
The same strategy for going straight can be applied to the zig-zag strip,
unfolding the appropriate crease to go left, straight, or right.
Thus we focus our case analysis on cases where the accordion is facing down.
Maneuvering the strip is more difficult in this case,
because the accordion naturally wants to unfold into the interior of~$P$,
which we forbid (and which could be a real problem because the surface of $P$
may have already been covered by previous portions of the strip).

\begin{figure}
\centering
\subfigure[]{\includegraphics[width=0.95\linewidth]{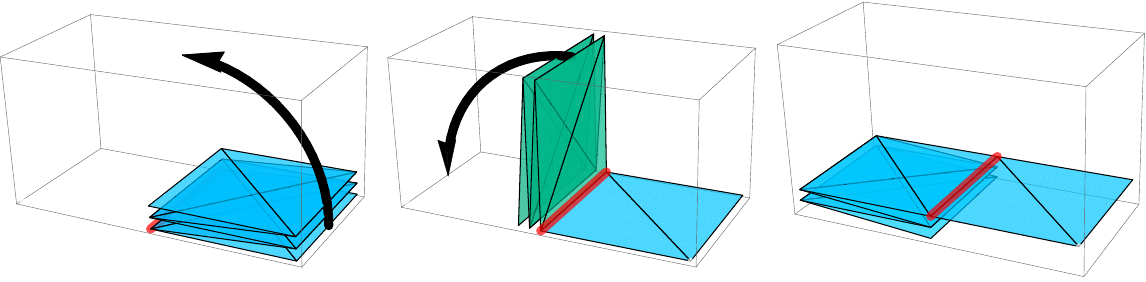}}

\subfigure[]{\includegraphics[width=0.95\linewidth]{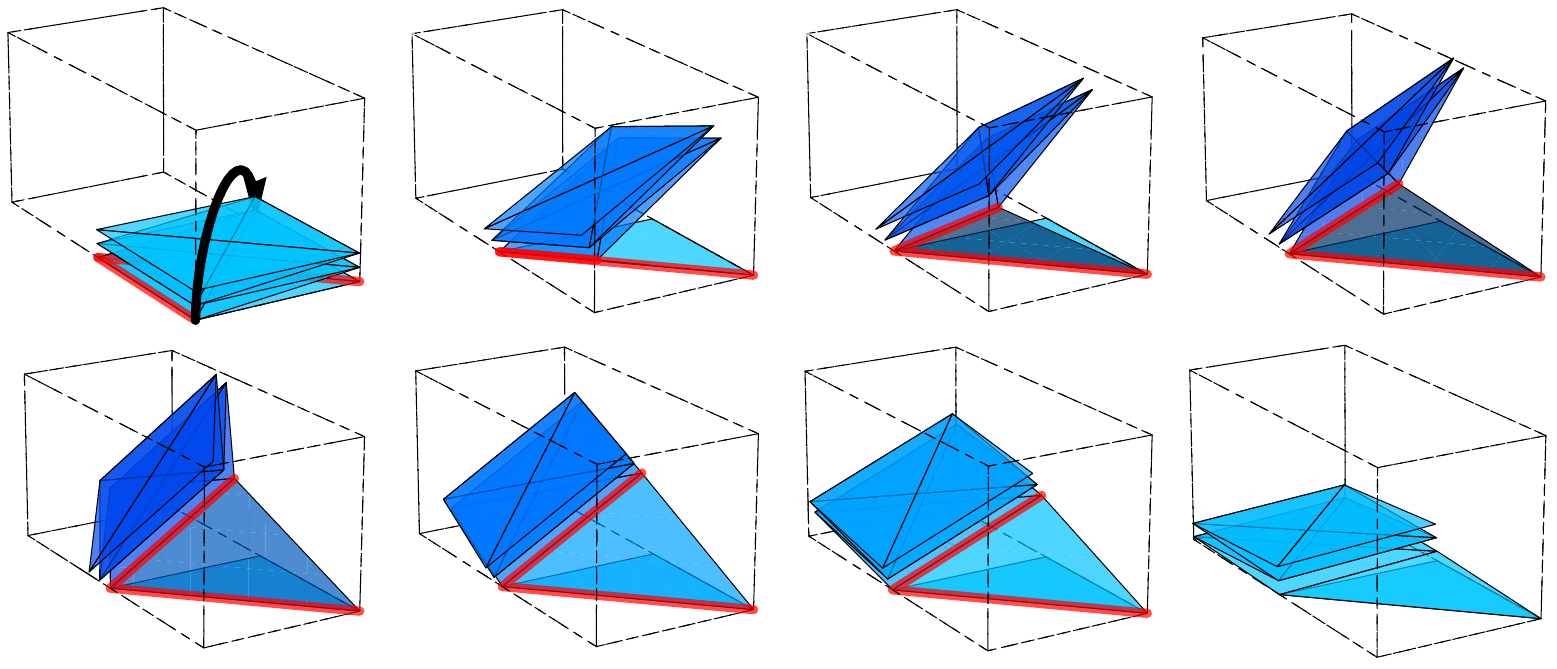}}
\caption{Continuous folding motion within a $1 \times 1 \times 2$ box for
canonical strip, face-up cases. (a) Straight. (b) Turn where $e_{23}$ is flat.}
\label{FaceUp}
\end{figure}


\subsection{Canonical Strip}
For the canonical strip, there are two turn cases and one straight case.

\paragraph{Turns.}
In the F--R turn case, $e_{12}$ is a flat edge and $e_{23}$ is a reflex edge,
as illustrated in Figure~\ref{StripFolding:fig:CanonicalTurnCase1Geometry}.
The moves we present for this case also work for any combination of $e_{12}$
being convex or flat and $e_{23}$ being reflex, flat, or convex.
Hence these moves cover all turn cases except for the R--R case
where $e_{12}$ and $e_{23}$ are both reflex;
see Figure~\ref{StripFolding:fig:CanonicalTurnCase2Geometry}.

Figures~\ref{StripFolding:fig:CanonicalTurnCase1Moves}
and~\ref{StripFolding:fig:CanonicalTurnCase2Moves} illustrate
the moves for the F--R and R--R turn cases, respectively.
The F--R turn case does not exploit feature size,
which is why the moves here still work for
any combination of $e_{12}$ being convex or flat
and $e_{23}$ being reflex, flat, or convex.
The R--R turn case uses that the dashed voxel in
Figure~\ref{StripFolding:fig:CanonicalTurnCase2} one unit away from $g_1$
must be empty because we assumed feature size at least~$2$,
and so the moves here rely on the specific position of $g_1$ relative to $g_2$
(i.e., they rely on $e_{12}$ being a reflex edge).

\begin{figure}
  \centering
  \subfigure[\label{StripFolding:fig:CanonicalTurnCase1Geometry}]
            {\includegraphics[scale=0.5]{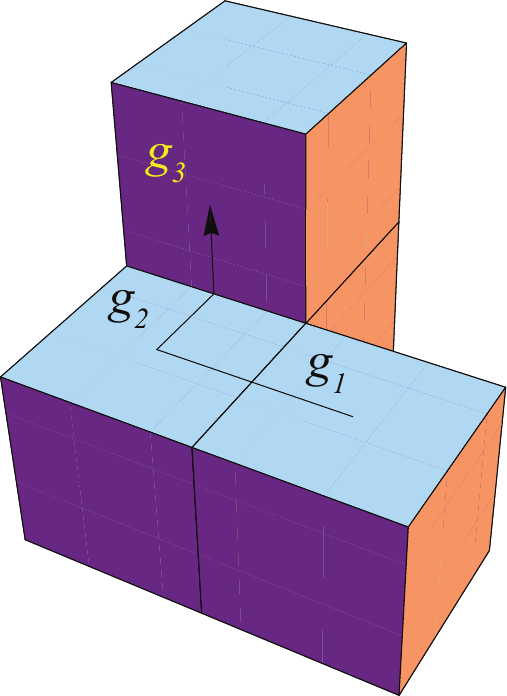}}\hfil\hfil\hfil
  \subfigure[\label{StripFolding:fig:CanonicalTurnCase1Moves}]
            {\includegraphics[scale=0.6]{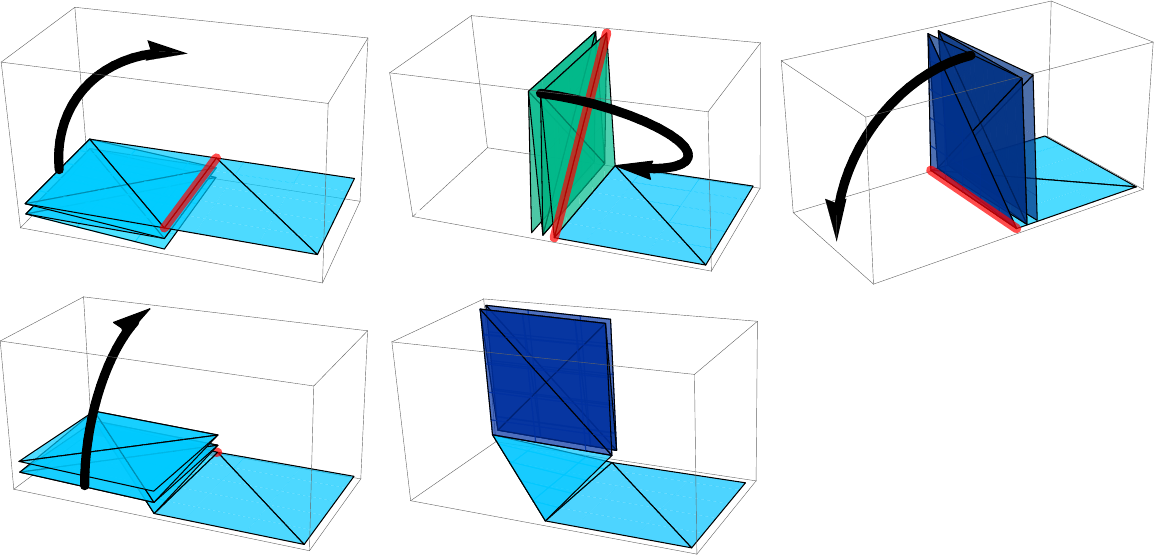}}
  \caption{Continuous folding motion for canonical strip, F--R turn case.
  (a) Local geometry. (b) Moves within a $1 \times 1 \times 2$ box
  for a face-down canonical accordion.}
  \label{StripFolding:fig:CanonicalTurnCase1}
\end{figure}

\begin{figure}
  \centering
  \subfigure[\label{StripFolding:fig:CanonicalTurnCase2Geometry}]
            {\includegraphics[scale=0.5]{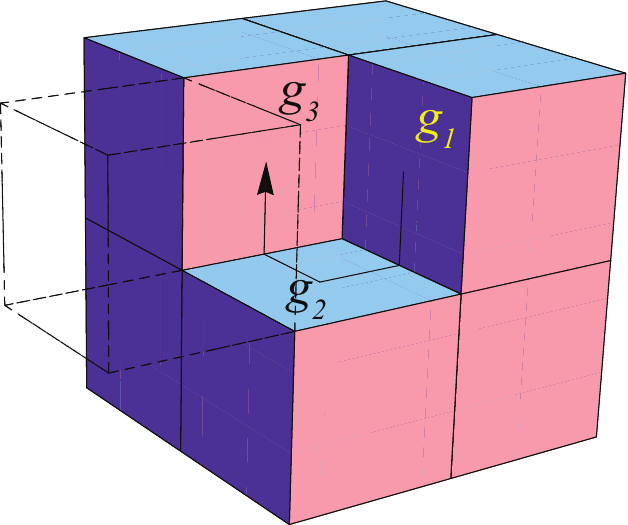}}\hfil\hfil\hfil
  \subfigure[\label{StripFolding:fig:CanonicalTurnCase2Moves}]
            {\includegraphics[scale=0.5]{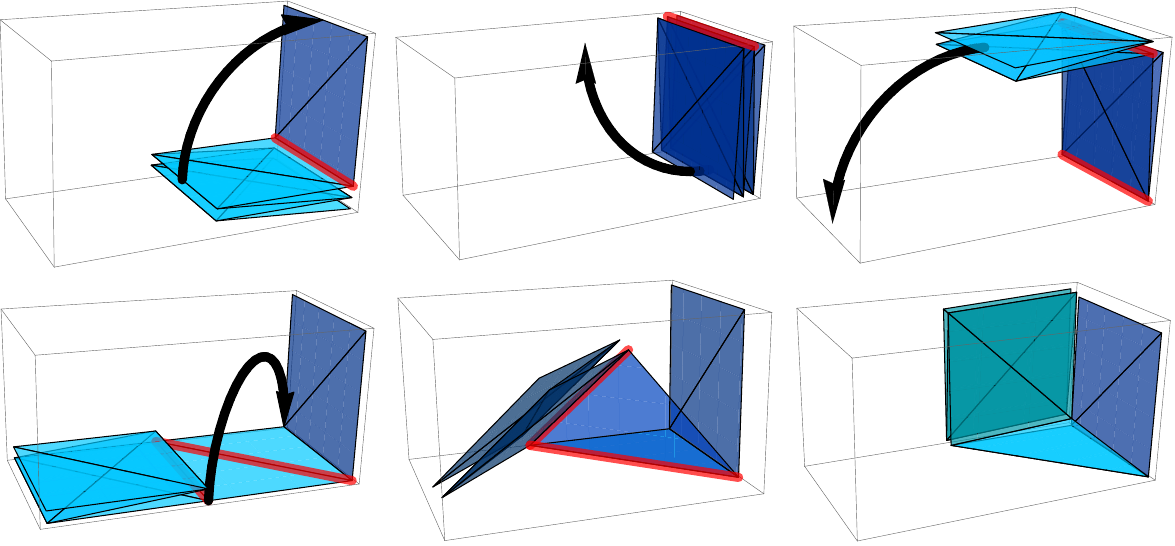}}
  \caption{Continuous folding motion for canonical strip, R--R turn case.
  (a) Local geometry. The dashed voxel beginning one unit away from $g_1$ must be empty because we assumed feature size at least~$2$.
  (b) Moves within a $1 \times 1 \times 2$ box for a face-down canonical accordion.}
  \label{StripFolding:fig:CanonicalTurnCase2}
\end{figure}

\paragraph{Straight.}
In the R--C straight case,
$e_{12}$ is a reflex edge and $e_{23}$ is a convex edge;
see Figure~\ref{StripFolding:fig:CanonicalStraightCase1Geometry}.
Figure~\ref{StripFolding:fig:CanonicalStraightCase1Moves} shows
the moves for this straight case where the canonical accordion lies face-down
on~$g_2$.
The moves for this case can be modified for any straight case
(i.e., any combination of edge types for $e_{12}$ and $e_{23}$)
except the R--R case when both $e_{12}$ and $e_{23}$ are reflex,
but this case is impossible while having the exterior voxel adjacent to $g_2$
in an empty $2 \times 2 \times 2$ box.

\begin{figure}
  \centering
  \subfigure[\label{StripFolding:fig:CanonicalStraightCase1Geometry}]
            {\includegraphics[scale=0.5]{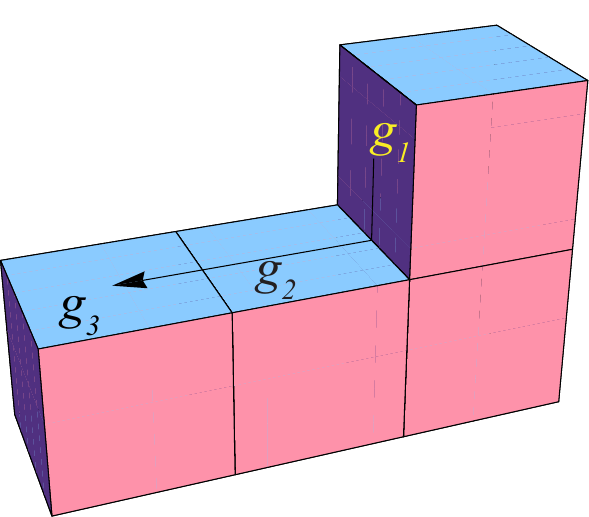}}\hfil\hfil\hfil
  \subfigure[\label{StripFolding:fig:CanonicalStraightCase1Moves}]
            {\includegraphics[scale=0.55]{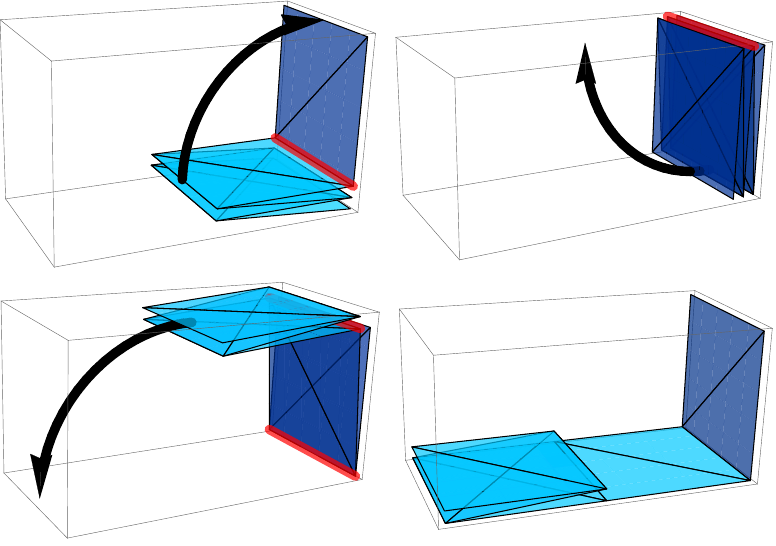}}
\caption{Continuous folding motion for canonical strip, R--C straight case.
(a) Local geometry.
(b) Moves within a $1 \times 1 \times 2$ box for a face-down canonical accordion.}
\label{StripFolding:fig:CanonicalStraightCase1}
\end{figure}

\subsection{Zig-Zag Strip}
The analysis for the zig-zag strip accordion moves is slightly more complicated
because we now also have to consider whether a turn costs $1$ unit or $3$ units,
as the moves can be different depending upon which turn-cost case we are in.

\paragraph{Turns.}
There are four turn cases for the zig-zag gadget that need to be considered:
\begin{description}
\item[F--R--$1$:] $e_{12}$ is flat, $e_{23}$ is reflex, and the turn costs~$1$;
\item[R--F--$*$:] $e_{12}$ is reflex, $e_{23}$ is flat, and the turn costs $1$ or~$3$
(the moves work regardless of the turn cost);
\item[R--R--$1$:] $e_{12}$ and $e_{23}$ are reflex, and the turn costs~$1$; and
\item[R--R--$3$:] $e_{12}$ and $e_{23}$ are reflex, and the turn costs~$3$.
\end{description}

We explain briefly why these are the only turn cases that need to be considered.
The moves for case F--R--1 also work when $e_{12}$ is convex and
$e_{23}$ is reflex with turn cost~$1$.
The moves for case R--F--$*$ also work for any combination of edge types
for $e_{12}$ and $e_{23}$ where $e_{23}$ is not reflex. 
The moves for case R--R--3 work when $e_{12}$ is flat or convex and
$e_{23}$ is reflex with turn cost~$3$.
Thus these cases collectively cover all possible cases we need to consider
for the zig-zag strip.

For brevity, we show the moves for just one of the turn cases, R--F--$*$.
Figure~\ref{StripFolding:fig:ZigZagTurnCase2} illustrates turn case R--F--$1$.
The same moves also work for turn cost~$3$. 

\begin{figure}
  \centering
  \subfigure[\label{StripFolding:fig:ZigZagTurnCase2Geometry}]
            {\includegraphics[scale=0.45]{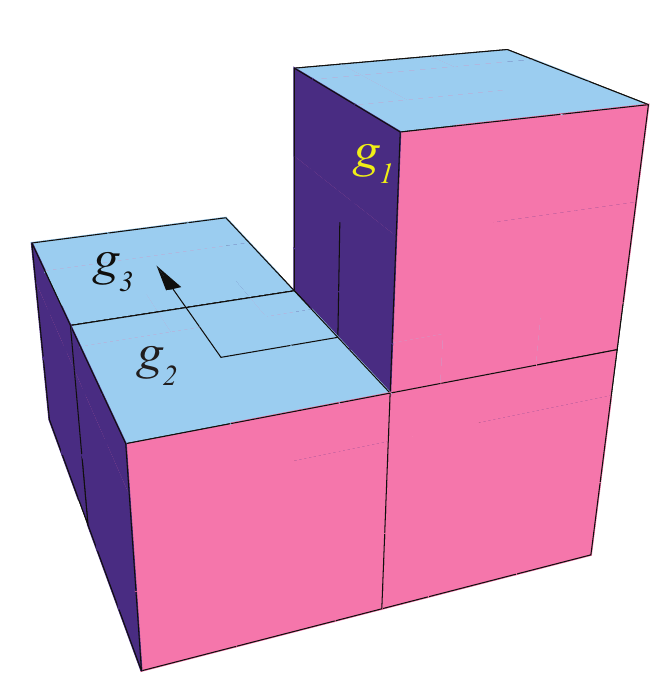}}\hfil\hfil\hfil
  \subfigure[\label{StripFolding:fig:ZigZagTurnCase2Moves}]
            {\includegraphics[scale=0.55]{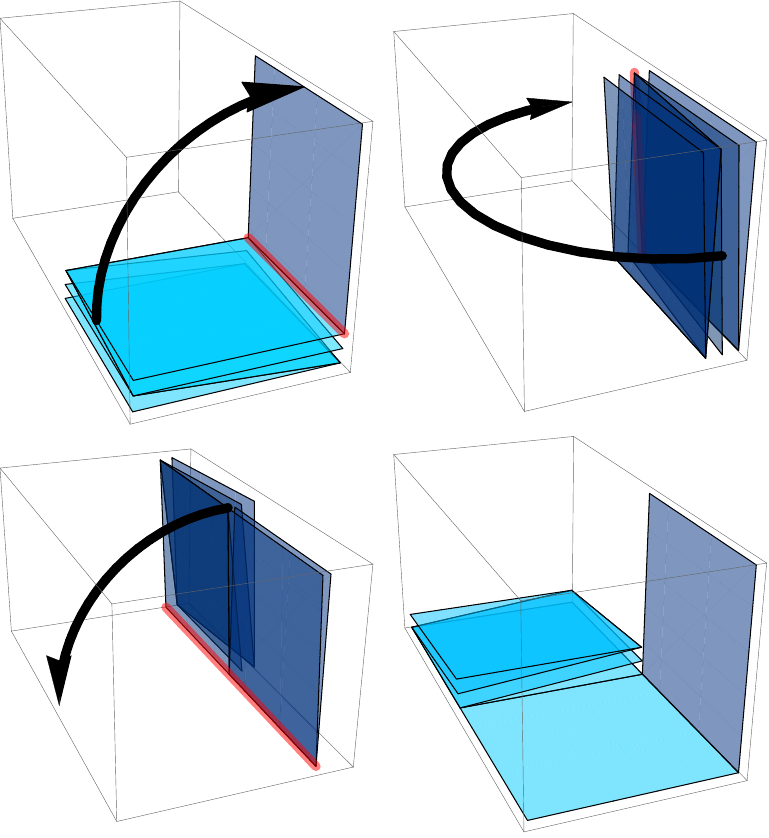}}
  \caption{Continuous folding motion for zig-zag strip, R--F--$1$ turn case. (a) Local geometry. (b) Moves within a $1 \times 1 \times 2$ box for a face-down zig-zag accordion where the turn
costs~$1$.  These moves also work for turn cost~$3$, and for any combination of edge types for
$e_{12}$ and $e_{23}$ not having $e_{23}$ reflex.}
  \label{StripFolding:fig:ZigZagTurnCase2}
\end{figure}

\paragraph{Straight.}
As with the canonical strip, it suffices to consider the R--C straight case
where $e_{12}$ is a reflex edge and $e_{23}$ is a convex edge.
The moves for this case with the zig-zag strip are analogous to those
in Figure~\ref{StripFolding:fig:CanonicalStraightCase1}
presented for the canonical strip, so we do not repeat them.

\section*{Acknowledgments}

We thank ByoungKwon An and Daniela Rus for several helpful discussions
about programmable matter that motivated this work.

\let\realbibitem=\bibitem
\def\bibitem{\par \vspace{-1.2ex}\realbibitem}

\bibliography{StripFolding}
\bibliographystyle{alpha}

\begin{figure}[b]
  \centering
  \includegraphics[width=\linewidth]{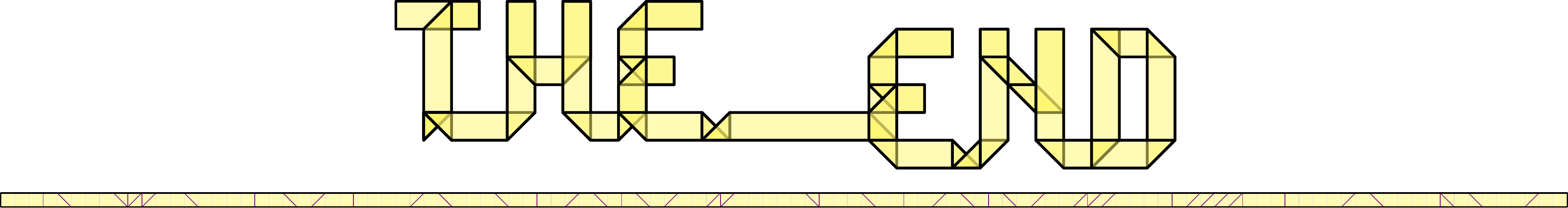}
  \caption{An example of joining together a few letters from our typeface
    in Figure~\ref{font}.  Unfolding (bottom) not to scale with folding
    (top).}
  \label{the end}
\end{figure}

\end{document}